%% file: arXiv.tex
\documentclass[11pt]{article}
\newif\ifcompact\compactfalse

\usepackage{bbm,amsmath}
\usepackage{float}

\usepackage{booktabs}
\usepackage{threeparttable}
\usepackage{multirow}

\usepackage[ruled,linesnumbered]{algorithm2e}

\usepackage{xcolor}

\usepackage{amsthm,amssymb}
\usepackage{natbib}
\usepackage{graphicx}
\usepackage{geometry}
\newtheorem{assumption}{Assumption}
\newtheorem{lemma}{Lemma}
\newtheorem{theorem}{Theorem}
\newtheorem{proposition}{Proposition}
\newtheorem{definition}{Definition}
\newtheorem{remark}{Remark}
\newtheorem{corollary}{Corollary}

\newenvironment{keywords}{{\centering\textbf{Keywords:}}}{}

\input{symbols.tex}

\def\target{Q}
\def\targetx{{\target_{X}}}
\def\targetxcy{\target_{X|Y = k}}
\def\targety{\target_{Y}}
\def\source{P_{j}}

\def\sourcexcy{P_{j,X|Y = k}}
\def\sourcexcyhat{\hat{P}_{j,X|Y = k}}
\def\targetxhat{\hat{Q}_{X}}
\def\sourcey{P_{j,Y}}

\def\w{\mathbf{w}}

\def\Lcal{\mathcal{L}}
\def\Dcal{\mathcal{D}}
\def\Ocal{\mathcal{O}}
\def\Xcal{\mathcal{X}}

\def\Ahat{\widehat{A}}
\def\bbar{\bar{b}}

\def\kernel{\mathcal{K}}
\def\betabold{\boldsymbol{\beta}}

\def\vqhat{\hat{\vq}}
\def\MMD{\mathrm{MMD}}
\def\MMDhat{\widehat{\MMD}}
\def\RW{\mathrm{RW}}
\def\RE{\mathrm{RE}}
\renewcommand{\hat}{\widehat}

\DeclareMathOperator*{\argmin}{\mathrm{arg\,min}}

\begin{document}

\title{Robust Multi-Source Domain Adaptation under Label Shift}

\author{
Congbin Xu$^a$, Chengde Qian$^b$, Zhaojun Wang$^a$ and Changliang Zou$^a$\\
{\small $^a$\it{School of Statistics and Data Science, Nankai University}}\\
{\small $^b$\it{School of Mathematical Sciences, Shanghai Jiao Tong University}}
}

\maketitle

\begin{abstract}
As the volume of data continues to expand, it becomes increasingly common for data to be aggregated from multiple sources. Leveraging multiple sources for model training typically achieves better predictive performance on test datasets. Unsupervised multi-source domain adaptation aims to predict labels of unlabeled samples in the target domain by using labeled samples from source domains. This work focuses on robust multi-source domain adaptation for multi-category classification problems against the heterogeneity of label shift and data contamination. We investigate a domain-weighted empirical risk minimization framework for robust estimation of the target domain's class proportion. Inspired by outlier detection techniques, we propose a refinement procedure within this framework. With the estimated class proportion, robust classifiers for the target domain can be constructed. Theoretically, we study the finite-sample error bounds of the domain-weighted empirical risk minimization and highlight the improvement of the refinement step. Numerical simulations and real-data applications demonstrate the superiority of the proposed method.
\end{abstract}


\section{Introduction}

Recent progress in multi-category classification has been greatly aided by the availability of large labeled datasets from the covariate-label space $\Xcal \times \Ycal$, where $\Xcal \subset \real^d$ and $\Ycal = [K] \triangleq \{1,\ldots,K\}$.
However, in many practical areas, such as autonomous vehicles and personalized healthcare, it is still very expensive and time-consuming to gather large-scale labeled datasets.
Therefore, there is a crucial need to apply knowledge from existing labeled datasets to make predictions in new, unlabeled situations.
This need has driven the development of \textit{unsupervised domain adaptation} \citep{kouw2019review,farahani2021brief}, which aims to build robust predictors for a target distribution $Q$ by using labeled samples from a source distribution $P$ and unlabeled samples from $Q$.
A key focus of this research is to effectively combine labeled and unlabeled data to improve predictive performance on target test data, especially when there is a significant \textit{distributional shift} between $P$ and $Q$.

\textit{Label shift}, a prevalent scenario in distributional shifts \citep{saerens2002adjusting,storkey2008training}, posits that while the marginal distribution of labels $Y$ may vary, the conditional distribution of covariates $X$ given $Y$ remains invariant.
Formally, a distribution $P$ is said to exhibit label shift if
\begin{equation}\label{pls}
	\targety \neq P_Y \text{ and } \targetxcy = P_{X|Y = k} \text{ for all } k \in \Ycal,
\end{equation}
with the notation $X|Y=k$ referring to the conditional distribution of $X$ given $Y=k$ under populations $Q$ and $P$.
A pivotal task in addressing label shift is the estimation of the target class proportion, denoted as $\vq^\ast = \{\Pr_\target(Y = k)\}_{k \in \Ycal}$.
With a precise estimate of $\vq^\ast$ in hand, the label shift can be effectively mitigated, paving the way for the development of high-accuracy predictors (classifiers) for the target domain \citep{zhang2013domain,lipton2018detecting}.
However, obtaining a reliable estimate of $\vq^\ast$ is inherently challenging within the unsupervised domain adaptation framework, as the labels from the target distribution $\target$ are unavailable.
To address this, several effective methodologies have been proposed.
Notably, the works of \cite{zhang2013domain}, \cite{iyer2014maximum}, \cite{gong2016domain}, and \cite{lipton2018detecting} have proposed systematic approaches to estimate $\vq^\ast$ in scenarios with a single source domain.
These approaches hinge on minimizing the divergence between the marginal distribution of the target domain $\targetx$ and a composite distribution derived from $\{P_{X|Y = k}\}$.

However, with increasing availability of data sources, applying the label shift assumption (\ref{pls}) across all sources becomes impractical, as
the cost of conducting thorough quality checks on each source is prohibitive.
For instance, in digital recognition tasks, the conditional distribution of $X$ given $Y = k$ can vary significantly between datasets like MNIST and SVHN, as highlighted by \cite{shui2021aggregating}.
In light of this, we address a scenario with $m$ sources $\{P_j\}$, each providing covariate-label pairs, alongside a target domain $Q$ consisting solely of unlabeled samples.
We propose a relaxed label shift assumption that allows an $\epsilon$ fraction of sources (denoted by $\Ocal \subset [m]$) to deviate from the standard label shift model.
Specifically, we adopt the following {\it multi-source contaminated label shift} framework:
\begin{equation}\label{equation:multiple_LS}
	\targety \neq \sourcey ,\, \targetxcy = \sourcexcy \text{ for all } k \in \Ycal, \, j \in [m]\setminus \Ocal,
\end{equation}
with $\abs{\Ocal} \leq \lfloor \epsilon m \rfloor$.
We refer to the sources in $\Ical \triangleq [m]\setminus \Ocal$ as inlier sources and those in $\Ocal$ as outlier sources.
Although the exact composition of $\Ical$ and $\Ocal$ remains unknown, we assume that the proportion $\epsilon$ (or an upper bound thereof) is known {\it a priori}.

It is clear that the single-source methods discussed earlier are not directly applicable in this context.
Indiscriminately combining all sources for model training may result in sub-optimal results, as outlier sources could introduce significant biases \citep{mansour2008domain,zhao2020multi}.
To develop better source aggregation strategies, \cite{li2019target} and \cite{shui2021aggregating} proposed methods to measure the similarity between source and target domains, thereby identifying potential outlier sources. \cite{li2019target} proposed a metric based solely on the marginal distributions of $X$ within a representational space, which fails to capture the underlying relationship in the context of label shift \citep{shui2021aggregating}.
In contrast, \cite{shui2021aggregating} investigated the conditional similarities of $X$ given $Y = k$, using labeled samples from $Q$, which are typically unavailable in the unsupervised domain adaptation paradigm.

This article addresses the critical challenge of developing a robust approach to handle the unsupervised label shift problem in the context of multiple sources, where some may not follow
the label shift assumption.
Our contributions are as follows:
\begin{itemize}
	\item We propose a robust estimation technique for $\vq^\ast$ by minimizing a weighted empirical divergence across sources.
	 The weights are assigned using robust algorithms that focus on excess risk, enhancing the identification of outlier sources due to the reduced variance in excess risk.
	 Leveraging the resulting estimators, we construct a robust classifier for the target domain by calibrating the label shifts.

	\item We employ the Maximum Mean Discrepancy (MMD) as a measure of distribution divergence and establish a convergent $\ell_2$-error bound for the proposed estimate.
	 This analysis reveals that the impact of outlier sources on the error bound is approximately proportional to $\sqrt{\epsilon/n}$, where $n$ represents the sample size of each source domain.

	\item We substantiate the effectiveness of our method through extensive numerical simulations and real-world dataset examples, showcasing its superiority over existing approaches.
\end{itemize}

\subsection{Related Work}
To calibrate general distribution shifts, the {\it likelihood-ratio-weighted risk minimization }aims to minimize weighted risk $\Rcal(\theta) = \E_{(X,Y) \sim P} [r(X,Y) \ell_\theta(X,Y)] $, where $\ell_\theta(\cdot,\cdot)$ is the loss function, $r(X,Y)$ is the likelihood ratio $\frac{d Q}{d P}(X, Y)$ between $\target$ and $P$; see, for example, \cite{shimodaira2000improving} and \cite{gretton2008covariate}.
In the context of the label shift, $r(X, Y)$ reduces to the class probability ratio $\frac{\Pr_\target(Y = k)}{\Pr_P(Y = k)}$ \citep{lipton2018detecting,azizzadenesheli2018regularized}.
Therefore, the key task is to estimate the target class proportion $\{\Pr_\target(Y = k)\}_{k \in \Ycal}$.
An extension is the generalized label shift setting, where the pair $(\phi(X), Y)$ obeys the label shift condition after an invariant feature mapping $\phi(X)$ \citep{gong2016domain,combes2020domain,wu2023prominent}.

{\it Transfer learning and domain adaption} also seek to build a predictor for the target population $\target$ using samples from the source distribution $P$ \citep{pan2010transfer}, but focuses on more intuitive techniques to bridge the gap between the varying distributions.
One widely used method is fine-tuning \citep{yosinski2014transferable}, where a predictor is first trained on source samples and then adjusted using a small set of target samples to improve performance on $\target$.
Another prominent approach is feature learning \citep{ganin2016dann}, which learns domain-invariant features using adversarial training techniques and then trains a predictor in the learned feature space with source samples.

{\it Federated learning} shares several similarities with multiple-source domain adaptation \citep{wen2023survey}. 
In federated learning, private data is stored locally across multiple clients.
Since clients are located in different areas, each client has heterogeneous data with varying distributions. 
Additionally, federated learning addresses 
data and model security concerns, including poisoning \citep{chan2018data} and Byzantine \citep{qian2024bymi} attacks, which pose threats by maliciously altering either the private data or the information transmitted during the learning process.
Despite these similarities, the objectives of federated learning and multiple-source domain adaptation diverge significantly. 
Federated learning aims to construct a global model that generalizes effectively across all local clients, while preserving data privacy by avoiding direct data transmission during training. 
Conversely, multiple-source domain adaptation emphasizes the development of a predictor specifically for the target domain, leveraging samples from multiple source domains to enhance its performance.

\subsection{Structure and Notation}

The remainder of this paper is structured as follows.
Section~\ref{sec:method} outlines the main ideas.
Section~\ref{sec:theory} delves into the theoretical guarantee of the proposed method.
Section~\ref{sec:numerical_study} provides simulation studies and real-data analyses.
Section~\ref{sec:conclusion} concludes the paper.
All theoretical proofs, together with additional algorithmic and numerical details, are provided in the Appendix.

We state the frequently used notation here.
For real-valued sequences $\{a_n\}$ and $\{b_n\}$, $a_n \lesssim b_n$ ($a_n \gtrsim b_n$) indicates that there is a positive constant $c$ such that $a_n < c b_n$ ($a_n > c b_n$) for all $n$.
For any positive integer $m$, let $[m] = \{1,\ldots,m\}$.
Denote the vector $(a_i)_{i \in [m]} = (a_1, \ldots, a_m)^\top \in \Rbb^{m}$ and the matrix $(a_{ij})_{i \in [m], j \in [n]} \in \Rbb^{m \times n}$ with the entry in the $i$-th row and $j$-th column be $a_{ij}$.
The $K$-dimensional probability simplex is represented as $ \triangle_{K} = \{\vv \in \real^K: v_k \geq 0,\, \sum_{k \in [K]} v_k = 1\} $. For a matrix $\mA$, $\lVert \mA \rVert_{\mathrm{op}}$ represents its operator norm.

\section{Methodology}\label{sec:method}

In this section, we propose a robust method for estimating $\vq^\ast$ and introduce the associated optimization method, followed by an example using the Maximum Mean Discrepancy as distribution divergence.

Suppose that we have labeled samples $\Dcal_{j} = \{(x_{j,i},y_{j,i})\}_{i \in [n]}$ from the $j$-th source and unlabeled sample $\Dcal_0 = \{x_i\}_{i \in [N]}$ from the target domain.
For the sake of simplicity in notation, we assume that $\Dcal_{j}$'s are in the same size but the methodology can fit general scenarios.
Denote $\Dcal_{j,k} = \{x_{j,i}: (x_{j,i},y_{j,i}) \in \Dcal_{j} \text{ and } y_{j,i} = k\}$ as the covariate set belonging to the $k$-th category and $n_{j,k} = \abs{\Dcal_{j,k}}$ as its size.

\subsection{Robust Estimation of Target Class Proportion}
Since $\targetx= \sum_k q_k^\ast \targetxcy$, the target class proportion $\vq^\ast$ minimizes
\[
	\Lcal_{\target}(\vq) = \Lcal\Bigl(\targetx, \sum_k q_k \targetxcy\Bigr),
\]
for an appropriate distribution divergence measurement $\Lcal$.
If all the source distributions $\source$'s satisfy the label shift assumption (\ref{pls}), we also have $\Lcal_{\target}(\vq)
 = \Lcal(\targetx, \sum_k q_k \sourcexcy)$
for $j\in [m]$, which can be approximated by the empirical loss values $\Lcal_{j}(\vq) = \Lcal_{j, \targetxhat}(\vq) \triangleq \Lcal(\targetxhat,\sum_{k} q_k \sourcexcyhat)$, where $\targetxhat$ and $\sourcexcyhat$ 
are the empirical distributions of the corresponding observed samples in the two domains.
Accordingly, we can directly obtain the estimate of $\vq^\ast$ by minimizing the joint objective $\sum_{j \in [m]}\Lcal_j(\vq)$.

However, under the scenario of multi-source contaminated label shift (\ref{equation:multiple_LS}), a reliable estimator is needed to ensure resilience against outlier sources.
We consider the following weighted empirical risk minimization framework as a robust proxy for $\Lcal_{\target}(\vq)$:
\begin{equation}\label{equation:robust_w_loss}
	\begin{aligned}
		\vqhat = \argmin_{\vq \in \triangle_{K}} \sum_{j \in [m]} w_{j} \Lcal_{j}(\vq) \text{ s.t. } \w = \RW_{\epsilon_h}(\{\Lcal_{j}(\vq) - \Lcal_{j}(\vq') \}),
	\end{aligned}
\end{equation}
where $\epsilon_h$ is some upper bound of the true contamination proportion $\epsilon$, and the robust weighting function $\RW_{\epsilon_h} $ returns a weight vector $\w = (w_{j})_{j \in [m]}$ such that $w_{j} \in \{0, (1 - \epsilon_h)^{-1}m^{-1}\}$ and $\sum_{j} w_{j} = 1$. 
For the sake of simplicity, we assume that $m\epsilon_h$ is an integer. Ideally, to eliminate the effect of outliers, the weights corresponding to outlier sources should all be zero, while the weights for inlier sources should remain positive.
Here $\vq'$ is a rough but robust estimate for $\vq^\ast$, and $\vqhat$ represents a one-step refinement of $\vq'$.
Similar one-step estimation method can be found in the literature of outlier detection \citep{cerioli2010multivariate,ro2015outlier}.
Practically, $\vq'$ can be set as the estimate obtained by
\begin{equation}\label{equation:robust_w_divergence}
	\begin{aligned}
		\vqhat_\Lcal = \argmin_{\vq \in \triangle_{K}} \sum_{j \in [m]} w_{j} \Lcal_{j}(\vq) \text{ s.t. } \w = & \RW_{\epsilon_h}(\{\Lcal_{j}(\vq)\}).
	\end{aligned}
\end{equation}
A multi-step refinement can also be implemented by updating $\vq'$ iteratively.

Various weighting-based methods in the literature of robust mean estimation or outlier detection can be
used as $\RW_{\epsilon_h}$.
To illustrate with an example,
we consider the weights obtained by removing an $\epsilon_h$ proportion of the input set to minimize the sample variance of the remaining observations, i.e.,
\begin{equation}\label{eq:m_weighted_var}
	\widehat{\Jcal} = \argmin_{\Jcal \subset [m], |\Jcal| = (1 - \epsilon_h) m} |\Jcal|^{-1} \sum_{j \in \Jcal} \Bigl(z_{j} - |\Jcal|^{-1}\sum_{j \in \Jcal} z_{j}\Bigr)^2,
\end{equation}
and $\RW_{\epsilon_h}(\{z_j\}) = \w = \{w_j\}_{j \in [m]} = \{|\widehat{\Jcal}|^{-1} \indicator_{\{j \in \widehat{\Jcal}\}}\}_{j \in [m]}$,
where $\{z_j\}_{j \in [m]} \subset \Rbb$ is the input set of $\RW_{\epsilon_h}$.
We name this method the minimum weighted variance (MWV) in this work.
Several multivariate robust mean estimators are essentially equivalent to MWV in univariate cases, including the filtering method \citep{diakonikolas2017being} and the minimum covariance determinant method \citep{butler1993asymptotics}.
Other robust methods, including the trimmed loss minimization \citep{rousseeuw1984least} and the median/median-of-means loss minimization \citep{lerasle2011robust, lugosi2019risk,lecue2020robustmachine}, can also be incorporated into our framework.
Additional discussions can be found in the Appendix.

\begin{figure}[tb]
	\centering
	\includegraphics[width=0.8\textwidth]{"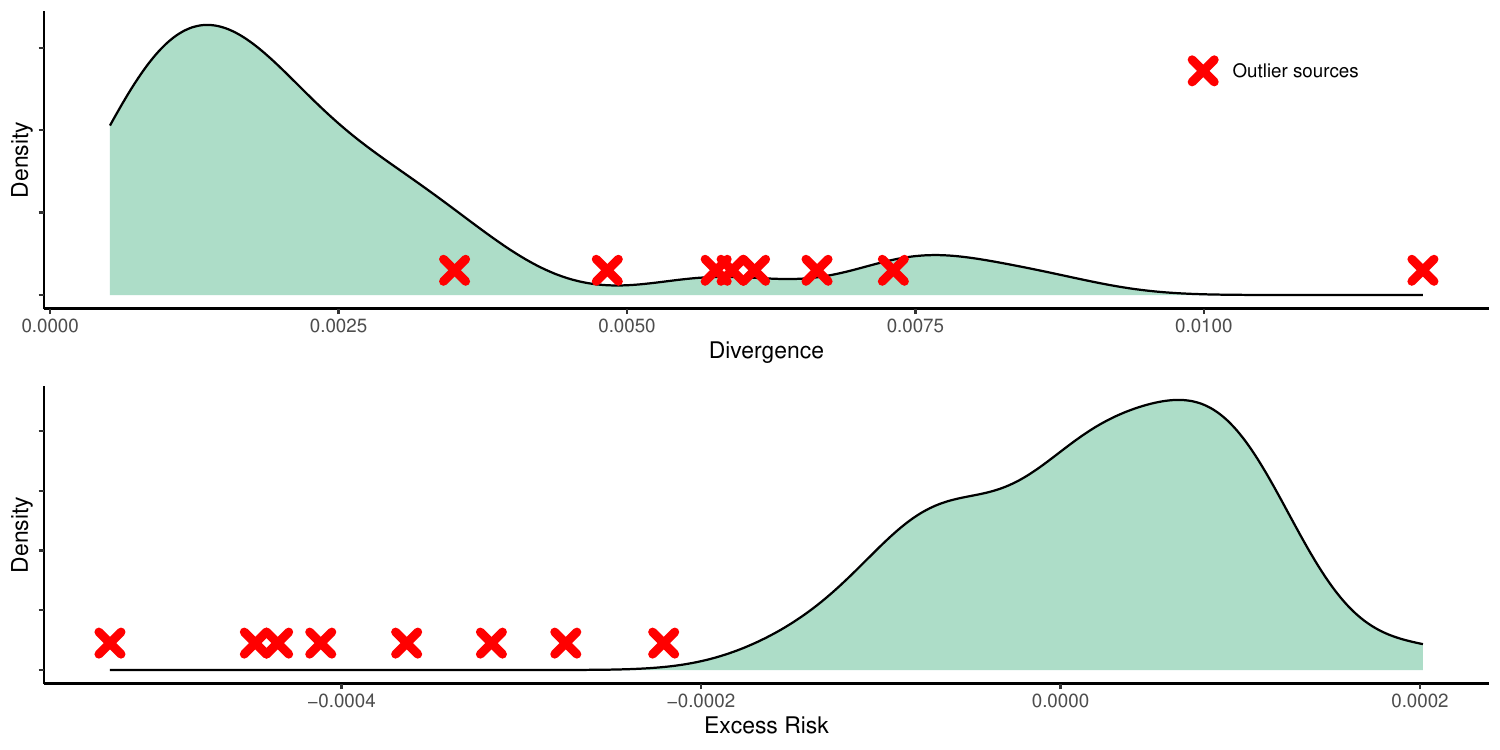"}
	\caption{The density plots of the inliers and the scatter plots of the outliers.
		We set $m = 40$, $n = 100$, $\epsilon = 20\%$ in the experiment.
		The details can be found in Section~\ref{sec:synthetic_data}.
	}
	\label{figure:divergence_vs_excess_risk}
\end{figure}

The robust weights in Eq.~(\ref{equation:robust_w_loss}) are obtained based on the excess risk $\Ecal_{j}(\vq,\vq') = \Lcal_{j}(\vq) - \Lcal_{j}(\vq')$ rather than the divergence $\Lcal_{j}(\vq)$, say the one in Eq.~(\ref{equation:robust_w_divergence}).
Using Eq.~(\ref{equation:robust_w_loss}) brings both theoretical and practical benefits, primarily due to the variance-reduction phenomenon of $\Ecal_{j}(\vq,\vq')$ comparing to $\Lcal_{j}(\vq)$ \citep{mathieu2021excess}.
Specifically, as $\vq'$ approaches $\vq^\ast$, the variance of $\Ecal_{j}(\vq, \vq')$ could be significantly smaller than that of $\Lcal_{j}(\vq)$.
In Section~\ref{sec:linear_model_variance_reduction}, a linear model is employed as an illustrative example to quantify the achieved variance reduction.
As demonstrated in Figure~\ref{figure:divergence_vs_excess_risk}, the smaller variance ensures that the excess risk of inlier sources is more concentrated in a smaller region, thereby making outlier sources easier to identify.
This advantage will be theoretically explored in Section~\ref{sec:theory}, where we show that our refined estimators $\hat{\vq}$ can achieve a faster convergence rate than the initial one $\vq'$.

\begin{remark}[Construction of robust classifiers] \label{remark:robust_classifier}
	Given $\vqhat$, we can further construct a classifier, parameterized by $\theta$, for the target domain $\target$ and estimate $\theta$ in a similar robust weighted empirical risk minimization framework with label shift calibration.
	Specially, we replace $\Lcal_j(\vq)$ and $\vq'$ in Eq.~(\ref{equation:robust_w_loss}) with the weighted empirical risk $\Rcal_j(\theta) = \sum_{i} \hat{q}_{y_{j,i}} / \hat{p}_{j,y_{j,i}} \ell_\theta(x_{j,i},y_{j,i})$ \citep{lipton2018detecting,azizzadenesheli2018regularized} and an initial estimate $\theta'$, respectively.
	Here, $\hat{\vp}_{j}$ is the empirical label proportion in the $j$-th source and $\ell_\theta(\cdot,\cdot)$ is some loss for classification tasks.
\end{remark}

\subsection{Optimization Algorithm}\label{sec:optimization}
Solving Eq.~(\ref{equation:robust_w_loss}) is a non-convex problem because the weight $\w$ varies with $\vq$ and lacks an explicit closed-form expression.
Therefore, we will focus on finding a local minimum.
A key observation is that almost surely (with respect to samples) for almost all $(\vq,\vq') \in \triangle_{K} \times \triangle_{K}$ (with respect to Lebesgue measure), there exists a neighbor around $(\vq,\vq')$ where $\RW_{\epsilon_h}(\{\Ecal_{j}(\vq,\vq')\})$ remains unchanged.
This implies that the derivative of $\w$ with respect to $\vq$ is zero except for a set of measure zero, almost surely.
For fixed $\{w_{j}\}$,
$ \sum_{j \in [m]} w_{j} \Lcal_{j}(\vq)$ is convex with repsect to $\vq$ when $\Lcal_j(\vq)$'s are convex.
Therefore, we can iteratively update $\vq$ and $\w$.
The gradient descent can be applied for updating $\vq$ with the gradient equals to $ \sum_{j \in [m]} w_{j} \nabla_{\vq} \Lcal_{j}(\vq)$, assuming $\Lcal_j(\vq)$'s are differentiable with respect to $\vq$.
We summarize the procedure in Algorithm~\ref{algorithm}.

\begin{algorithm}[tb]
	\label{algorithm}
	\KwIn{Empirical divergence $\{\Lcal_j(\vq)\}$, the hyperparameter $ \epsilon_h $, a rough estimate $\vq'$, a randomized initial value $\vq_0$, the maximum number of iterations $ T $, and the learning rate sequence $ \{\gamma_t\} $}
	\KwOut{Label proportion estimator $ \vqhat $ }
	\For{$ t = 1,\ldots,T $}{
	Update $ \w $: $\w_t = \RW_{\epsilon_h}( \{\Lcal_{j}(\vq_{t-1}) - \Lcal_{j}(\vq')\}_{j \in [m]} )$\\
	Update $ \vq $: $ \vq_{t} = \operatorname{Prox}_{\triangle_{K}}(\vq_{t-1} - \gamma_t \sum_{j \in [m]}w_{t,j} \nabla_{\vq} \Lcal_{j}(\vq_{t-1})) $, where $ \operatorname{Prox}_{\triangle_{K}}(\cdot) $ is projection operator onto $ \triangle_{K} $
	}
	\Return{$ \vqhat = \vq_{T} $}
	\caption{Optimization procedure for Eq.~(\ref{equation:robust_w_loss})}
\end{algorithm}

\subsection{The Implementation via Maximum Mean Discrepancy}

We use the MMD for measuring the divergence between $\targetx$ and $\sum_k q_k \sourcexcy$, which is a frequently used measure and widely adopted in domain adaptation tasks \citep{pan2010domain,long2013transfer}.
We refer to \cite{gretton2012kernel} for a detailed introduction.

The MMD between the marginal $\targetx$ and the mixed distribution $\sum_{k=1}^{K} q_k \targetxcy$, up to a constant not depending on $\vq$, can be represented as
\begin{align}\label{equation:source_population_MMD}
	\MMD\Bigl(\targetx,\sum_{k=1}^{K} q_k \targetxcy\Bigr) 
	\propto \vq^\top \mA \vq - 2 \vq^\top \vb ,
\end{align}
where 
$ \mA = (A_{k,k^\prime})_{k,k^\prime = 1}^{K} \in \real^{K\times K} $, $ \vb = (b_k)_{k=1}^{K} \in \real^K $ are defined as
\begin{equation*}
	\begin{aligned}
		A_{k,k^\prime} & = \E_{x, x^\prime} [\kernel(x,x^\prime)], \quad x \sim \targetxcy,\ x^\prime \sim \target_{X|Y = k^\prime}, \\
		b_k & = \E_{x, x^\prime} [\kernel(x,x^\prime)] , \quad x \sim \targetxcy,\ x^\prime \sim \targetx,
	\end{aligned}
\end{equation*}
for independent $ x $ and $ x^\prime $, and $\kernel: \Xcal\times \Xcal \rightarrow \real$ is a symetric kernel.
Under suitable conditions, $\MMD(\targetx,\sum_k q_k \targetxcy) = 0$ if and only if $\targetx = \sum_k q_k \targetxcy$.
Consequently,
we can use samples to estimate $\mA$ and $\vb$ \citep{gretton2012kernel,yan2017mind}, and obtain an unbiased estimate of the MMD in (\ref{equation:source_population_MMD}):
\begin{equation}\label{equation:MMD}
	\MMDhat\bigl(\vq;\Dcal_0,\Dcal_{j}\bigr) = \vq^\top \hat{\mA}_{j} \vq - 2 \vq^\top \hat{\vb}_{j}, 
\end{equation}
where $\hat{\mA}_{j} = (\hat{A}_{j,k,k^\prime})_{k,k^\prime = 1}^K$, $\hat{\vb}_{j} = (\hat{b}_{j,k})_{k = 1}^K$,
$\hat{A}_{j,k,k} = \frac{1}{n_{j,k} (n_{j,k} - 1)} \sum_{x, x^\prime \in \Dcal_{j,k}, x \neq x^\prime}\kernel(x,x^\prime)$,
$\hat{A}_{j,k,k^\prime} = \frac{1}{n_{j,k} n_{j,k^\prime}} \sum_{x \in \Dcal_{j,k}} \sum_{ x^\prime \in \Dcal_{j,k^\prime}} \kernel(x,x^\prime)$, and $\hat{b}_{j,k} = \frac{1}{n_{j,k} N} \sum_{x \in \Dcal_{j,k}} \sum_{x^\prime \in \Dcal_0}\kernel(x,x^\prime)$.
Algorithm~\ref{algorithm} can be applied with $\Lcal_{j}(\vq)$ replaced by $\MMDhat (\vq;\Dcal_0,\Dcal_{j})$.

\section{Theory}\label{sec:theory}
In this section, under the multi-source contaminated label shift model in Eq.~(\ref{equation:multiple_LS}), we first present theoretical results for the proposed method with general distribution divergence and robust weighting functions.
Subsequently, by specifying the distribution divergence as MMD and the robust weighting method as MWV, we derive the specific convergence rate of the proposed method.

Denote $\Ecal_{\target}(\vq,\vq^\prime) = \Lcal_{\target}(\vq) - \Lcal_{\target}(\vq^\prime)$ as the expected excess risk.
For general divergence and robust weighting methods, we consider the following assumptions.
\begin{assumption}\label{ass:convex}
	The empirical loss functions $\{\Lcal_j(\vq)\}$ are differentiable, and that the population loss $\Lcal_{\target}(\vq)$ satisfies $ \Lcal_\target(\vq) - \Lcal_\target(\vq^\ast) \geq \lambda \norm{\vq - \vq^\ast}_2^{2}$ for some $\lambda > 0$ and $\forall \vq \in \triangle_{K} $.
\end{assumption}

Assumption~\ref{ass:convex} ensures that the gradients of the empirical losses are well-defined and the population loss $\Lcal_{Q}(\vq)$ is strongly convex. The strong convexity ensures that $\vq^\ast$ is the unique minimizer of $\Lcal_\target(\vq)$.

\begin{assumption}\label{ass:uniform_converge_excess_risk}
	There exists a sequence $r_n = o(1)$ such that the robust excess risk estimation $\hat{\Ecal}(\vq_1,\vq_2) = \RE_{\epsilon_h} (\{\Ecal_{j}(\vq_1,\vq_2)\}_{j \in [m]})$, defined as the weighted mean estimation with the weight returned by $\RW_{\epsilon_h} (\{\Ecal_{j}(\vq_1,\vq_2)\}_{j \in [m]})$, satisfies that with probability at least $1 - \tau$,
	\[
		\sup_{\vq_1,\vq_2 \in \triangle_{K}} \frac{\abs{\hat{\Ecal}(\vq_1,\vq_2) - \Ecal_Q(\vq_1,\vq_2)}}{\max(\norm{\vq_1 - \vq^\ast}_{2} + \norm{\vq_2 - \vq^\ast}_{2}, r_n)} \lesssim r_n.
	\]
\end{assumption}

In Assumption~\ref{ass:uniform_converge_excess_risk}, the inclusion of the maximum in the denominator serves to prevent the possibility of division by zero. This assumption primarily asserts that the mean estimators derived from the robust weighting method can robustly and uniformly approximate the target $\Ecal_Q(\vq_1, \vq_2)$ for all pairs $(\vq_1, \vq_2)$. This uniform convergence is crucial for deriving the convergence rate of the proposed robust weighted risk minimization method.

Treat $\RW_{\epsilon_h}(\{\Ecal_{j}(\vq, \vq^\prime\})$ as a function of $\vq$, i.e., $\w(\vq) = \RW_{\epsilon_h}(\{\Ecal_{j}(\vq, \vq^\prime\})$.
In this context, we impose an additional identifiability assumption on the weighting function $\w(\vq)$.

\begin{assumption}\label{ass:weight_identifiability}
	There exists a constant $C > 0$ such that $\w(\vq)$ remains unchanged in the region $\{\vq: \lVert \vq - \vq^\ast \rVert_2 \le C \sqrt{r_n}\}$.
\end{assumption}

Given that $r_n = o(1)$ and the weights $\{w_{j} \in \{0, (1 - \epsilon_h)^{-1}m^{-1}\} \}$ are discrete, Assumption~\ref{ass:weight_identifiability} ensures that the identified outlier set $\{j \in [m]: w_j(\vq) = 0\}$ remains consistent within the small region $\{\vq: \lVert \vq - \vq^\ast \rVert_2 \le C \sqrt{r_n}\}$.
This assumption is generally reasonable, as it facilitates stability in the identification of outliers despite the inherent variability in the data.
When employing a continuous loss function $\Lcal_j(\vq)$, such as the MMD loss, in conjunction with a standard weighting method like the MWV, the weight function $\w(\vq)$ is characterized as piecewise constant within the domain $\triangle_{m}$. 
Furthermore, it remains locally unchanged almost everywhere with respect to the Lebesgue measure.

Building upon the three assumptions delineated, we present the following theorem for the proposed framework as specified in Eq.~(\ref{equation:robust_w_loss}) with general distribution divergence and robust weighting methods.
\begin{theorem}\label{theorem:general_one_step_update}
	Under Assumptions~\ref{ass:convex}, \ref{ass:uniform_converge_excess_risk} and \ref{ass:weight_identifiability}, if the rough estimate $\vq^\prime$ has an error bound $\norm{\vq^\prime - \vq^\ast}_{2} \leq r_n^\prime$, then there exists a local minimizer $\vqhat$ of Eq.~(\ref{equation:robust_w_loss}) that satisfies the following error bound with probability at least $1 - \tau$,
	\begin{equation}\label{eq:err_onestep_general}
		\norm{\vqhat - \vq^\ast}_{2} \lesssim \max\{r_n,\sqrt{r_n r_n^\prime}\}.
	\end{equation}
\end{theorem}

Theorem~\ref{theorem:general_one_step_update} establishes a connection between the error bound $O(r_n)$ of the robust weighted mean estimator for the excess risk as articulated in Assumption~\ref{ass:uniform_converge_excess_risk} and the error bound of the estimator derived from the proposed robust weighted risk minimization framework in Eq.~(\ref{equation:robust_w_loss}).
This theorem indicates that once the error rate $r_n^\prime$ associated with the rough estimate $\vq^\prime$ exceeds $r_n$, a one-step refinement process will reduce the estimation error.
Building upon this result, we subsequently present a corollary regarding the multi-step refinement estimator.

\begin{corollary}\label{corollary:general_multistep_update}
	Under the same assumptions as in Theorem~\ref{theorem:general_one_step_update}, consider the multi-step refinement procedure, where the $t$-th estimate $\vqhat_t$ is obtained by setting $\vq^\prime=\vqhat_{t-1}$ in Eq.~(\ref{equation:robust_w_loss}).
	If the initial estimate $\vqhat_0$ satisfies $\norm{\vqhat_0 - \vq^\ast}_2 \leq r_n^\prime$ with $r_n^\prime \gtrsim r_n$ and $\vqhat_t$'s are local minimizers described in Theorem~\ref{theorem:general_one_step_update}, then with probability at least $1-\tau$, we have
    \begin{equation}\label{eq:err_multistep_general}
        \norm{\vqhat_{t} - \vq^\ast}_2 \lesssim r_n^{1-2^{-t}} (r_n^\prime)^{2^{-t}} \mbox{and}\, \lim_{t \rightarrow \infty} \norm{\vqhat_{t} - \vq^\ast}_2 \lesssim r_n.
    \end{equation}
\end{corollary}

As the number of steps $t$ in Corollary~\ref{corollary:general_multistep_update} increases, the error rate of the multi-step refinement estimator ultimately reduces to $O(r_n)$.

To illustrate the enhancement of the robust weighting method based on the excess risk values in comparison to the conventional approach that relies on loss values, we introduce Assumption~\ref{ass:uniform_converge_divergence}, which controls the error in robust loss estimation.

\begin{assumption}\label{ass:uniform_converge_divergence}
	The $\hat{\Lcal}(\vq) = \RE_{\epsilon_h} (\{\Lcal_{j}(\vq)\}_{j \in [m]})$ satisfies that with probability at least $1 - \tau$,
	\[\sup_{\vq \in \triangle_{K}} | \hat{\Lcal}(\vq) - \Lcal_\target(\vq)| \lesssim r_n,\]
    with the same sequence of error rate $r_n$ in Assumptions~\ref{ass:uniform_converge_excess_risk}.
\end{assumption}
We demonstrate the equivalence of the error rate $r_n$ in Assumptions~\ref{ass:uniform_converge_excess_risk} and \ref{ass:uniform_converge_divergence} specifically using the MWV method as an example.
By substituting Assumption~\ref{ass:uniform_converge_excess_risk} with Assumption~\ref{ass:uniform_converge_divergence}, we can derive the error bound of the estimator $\vqhat_\Lcal$ defined in Eq.~(\ref{equation:robust_w_divergence}).

\begin{theorem}\label{theorem:general_robust_loss}
	Under Assumptions~\ref{ass:convex} and \ref{ass:uniform_converge_divergence}, the minimizer $\vqhat_\Lcal$ in Eq.~(\ref{equation:robust_w_divergence}) satisfies the following error bound with probability at least $1 - \tau$,
	\begin{equation}\label{eq:err_minloss_general}
	    \norm{\vqhat_\Lcal - \vq^\ast}_{2} \lesssim \sqrt{r_n}.
	\end{equation}
\end{theorem}

By Theorem~\ref{theorem:general_robust_loss}, if we choose $\vqhat_0 = \vqhat_\Lcal$ in Corollary~\ref{corollary:general_multistep_update}, then $r_n^\prime = O(\sqrt{r_n})$.
For the $t$-th iteration ($t \ge 1$), the error bound of $\vqhat_t$ becomes $O(r_n^{1 - 2^{-t-1}})$, which is considerably faster than the error bound of $\vqhat_{\Lcal}$, $O(\sqrt{r_n})$, in Theorem~\ref{theorem:general_robust_loss}.

By designating the loss function as MMD and the robust weighting function as MWV, we determine the error rate $r_n$ for the proposed method, thereby highlighting its robustness and efficiency. 
For the analysis, we need the following assumptions regarding the data-generating process and the properties of the MMD loss.
\begin{assumption}\label{ass:sample_balance}
	For inlier sources, $ \min_{k \in [K], j \in \Ical} n_{j,k} \geq n\psi > 0 $ for some constant $\psi > 0$.
\end{assumption}

Assumption~\ref{ass:sample_balance} posits that the sample sizes across all categories and domains are balanced, which is a criterion frequently employed in the literature \citep{lipton2018detecting,azizzadenesheli2018regularized}.
This assumption serves a technical purpose, enabling a streamlined discussion and proof of the subsequent results.

\begin{assumption}\label{ass:linearly_independent}
	Under the label shift model, for the target domain and the inlier domains $ j \in \Ical$, the conditional distributions $\{\targetxcy\}_{k \in [K]} = \{\sourcexcy\}_{k \in [K]}$, are linearly independent.
\end{assumption}

Assumption~\ref{ass:linearly_independent} is also employed in \cite{garg2020unified} to guarantee the identifiability of $\vq^\ast$.
It is important to note that Assumptions~\ref{ass:sample_balance}--\ref{ass:linearly_independent} are imposed on inlier sources, while no requirements are placed on outlier sources.
The data in the outlier sources can follow arbitrary distribution and may violate the label shift assumption.

\begin{assumption}\label{ass:kernel}
	(i) The kernel $\Kcal$ is symmetric, positive and bounded above by $\kappa > 0$, i.e., $\kernel(x_1,x_2) = \kernel(x_2,x_1) \leq \kappa$ and $\sum_{i,j \in [s]} a_{i} a_{j}\kernel(x_i,x_j) > 0$ for any positive integer $s$ and vector $(a_i)_{i \in [s]} \in \real^s \setminus \{\boldsymbol{0}\}$, and $\{x_i\}_{i \in [s]} \subset \Xcal$.
	(ii) The kernel $\kernel$ is characteristic, i.e., $\kernel(x_1,\cdot) \neq \kernel(x_2,\cdot)$ for distinct $x_1, x_2 \in \Xcal $.
	(iii) $\max_{k \in [K]} \lVert \Cov(\phi_k) \rVert_{\mathrm{op}} < C$ for some constant $C$ where $\phi_k = (\Kcal(x_{k1}, x_{g2}))_{g \in [K]}$, and $\{x_{k1}, x_{k2}\}$ are i.i.d. samples from $\targetxcy$.
\end{assumption}

The characteristic kernel assumption on $\Kcal$ guarantees that $\MMD(\targetx,\sum_k q_k \targetxcy) = 0$ if and only if $\targetx = \sum_k q_k \targetxcy$ \citep{fukumizu2007kernel,gretton2012kernel}.
Many commonly used kernels, e.g.,
the Gaussian kernel $\Gcal(x_1,x_2)= \exp\{-{\norm{x_1 - x_2}_2^2}/(2\sigma^2)\}$, satisfy Assumption~\ref{ass:kernel} \citep{sriperumbudur2008injective}.
Assumptions~\ref{ass:linearly_independent} and \ref{ass:kernel} together guarantee that the matrix $\mA$ in Eq.~(\ref{equation:source_population_MMD}) is positive definite, which further ensures Assumption~\ref{ass:convex}.

With the above assumptions, we now focus on determining the rate $r_n$. 
Specifically, consider samples $\{ z_j \}_{j=1}^m$ that are independently generated from distributions with a common mean and variances upper bounded by $\sigma_z^2$. If a proportion $\epsilon$  of these samples is arbitrarily corrupted, any robust mean estimate derived from $\{ z_j \}_{j=1}^m$ will maintain a minimax lower bound on the error rate of $\sigma_z \sqrt{\epsilon}$ \citep{catoni2012challenging, chen2016general}. It is important to note that the empirical loss values $\{ \Lcal_{j}(\vq) = \Lcal_{j, \targetxhat}(\vq) = \Lcal(\targetxhat,\sum_{k} q_k \sourcexcyhat)\}$ are correlated, as they derive from the same empirical distribution $\targetxhat$, which similarly affects the excess risk values $\{ \Ecal_{j}(\vq_1,\vq_2)\}$. 
To align with the standard independence setting, we can instead investigate the values $\{ \Lcal_{j, \targetx}(\vq) = \Lcal(\targetx,\sum_{k} q_k \sourcexcyhat)\}$ and $\{ \Ecal_{j, \targetx}(\vq_1, \vq_2) = \Lcal_{j, \targetx}(\vq_1) - \Lcal_{j, \targetx}(\vq_2)\}$, which are accurate approximations for $\{ \Lcal_{j}(\vq)\}$ and $\{ \Ecal_{j}(\vq_1,\vq_2)\}$, respectively. Subsequently, the properties of the empirical ones can be obtained.

Based on the above guidance, to determine the rate $r_n$ specified in Assumptions~\ref{ass:uniform_converge_excess_risk} and \ref{ass:uniform_converge_divergence}, it is essential to derive the upper bounds of the variances associated with $\{ \Ecal_{j, \targetx}(\vq_1, \vq_2)\}$ and $\{\Lcal_{j, \targetx}(\vq)\}$.The related findings are presented in Lemma~\ref{lemma:upper_var}.
\begin{lemma}\label{lemma:upper_var}
    Under Assumptions~\ref{ass:sample_balance} and \ref{ass:kernel}, for some universal constant $C > 0$, it holds that
     \begin{equation*}
		\Var[\Ecal_{j,\targetx}(\vq_1,\vq_2)] \le \frac{C \norm{\vq_1 - \vq_2}_2^2}{n \psi} \le \frac{C (\norm{\vq_1-\vq^\ast}_{2} + \norm{\vq_2-\vq^\ast}_{2})^2}{n\psi},
	\end{equation*}
    uniformly for all $j \in [m]$ and $(\vq_1, \vq_2) \in \Delta_{K} \times \Delta_{K}$. Similarly, it holds that
    \begin{equation*}
		\Var[\Lcal_{j,\targetx}(\vq)] \le \frac{C}{n\psi},
	\end{equation*}
    uniformly for all $j \in [m]$ and $\vq \in \Delta_K$.
\end{lemma}

From Lemma~\ref{lemma:upper_var}, we can observe the variance reduction of $\Var[\Ecal_{j,\targetx}(\vq_1,\vq_2)]$ compared to $\Var[\Lcal_{j,\targetx}(\vq)]$, in the sense that $\Var[\Ecal_{j,\targetx}(\vq_1,\vq_2)] = o(\Var[\Lcal_{j,\targetx}(\vq)])$ when $\norm{\vq_1 - \vq_2}_2 = o(1)$.
With the variances controlled, we can establish the following uniform results for the MWV estimator.
\begin{lemma} \label{lemma:uniform_bound_mwv}
    When using the MWV method with $\epsilon_h$ slightly larger than $\epsilon$ and $\epsilon_h = O(\epsilon)$ to estimate $\hat{\Ecal}(\vq_1, \vq_2)$ and $\hat{\Lcal}(\vq)$, the excess risk estimator satisfies that with probability at least $1 - \tau$,
    \begin{equation*}
        \sup_{\vq_1,\vq_2 \in \triangle_{K}} \frac{\abs{\hat{\Ecal}(\vq_1,\vq_2) - \Ecal_{\target}(\vq_1,\vq_2)}}{\max(\norm{\vq_1 - \vq^\ast}_{2} + \norm{\vq_2 - \vq^\ast}_{2},r_n)}  \lesssim r_n,
    \end{equation*}
    and the loss estimator satisfies that with probability at least $1 - \tau$,
    \begin{equation*}
        \sup_{\vq \in \triangle_{K}} \abs{\hat{\Lcal}(\vq) - \Lcal_\target(\vq)} \lesssim r_n,
    \end{equation*}
        where 
        \begin{equation}\label{eq:specified_rn}
            r_n = \sqrt{\frac{\epsilon}{n \psi}} + \sqrt{\frac{\log(1/\tau) + K \log (K m n) }{mn \psi}} + \kappa\sqrt{\frac{K\log(mK/\tau)}{N}}.
        \end{equation}
\end{lemma}

The proof of Lemma~\ref{lemma:uniform_bound_mwv} primarily hinges on establishing the estimation accuracy of the MWV estimator, leveraging the stability arguments developed by \citet{diakonikolas2020outlier}. 
In contrast to the conventional robust mean estimation problem, a significant challenge in the proof arises from the necessity for the error bounds to hold uniformly for all pairs $(\vq_1, \vq_2)$ and all $\vq$.
Further details regarding the proof can be found in the Appendix.

The rate $r_n$ is comprised of three distinct terms.
The first term, which emerges from contamination, signifies that the most detrimental effect of the contamination is of rate $ \sqrt{{\epsilon}/{n}}$. 
This aligns with the optimal dimension-free error bound in distributed robust learning tasks, as noted by \cite{zhu2023byzantine}. 
The second term, $ \sqrt{{1}/({nm})} $, achieves, up to a logarithmic factor, the optimal convergence rate for aggregating information from inlier sources.
The third term is associated with the convergence rate of $\targetxhat$ to $\targetx$. Given the prevalent assumption that the sample size of unlabeled data $ N $ is generally large, as in the context of semi-supervised learning, this term is often negligible compared to the others. 
Consequently, when $\epsilon$ is treated as a constant, the leading term of the error bound can be predominantly represented by $\sqrt{{\epsilon}/{n}}$.

We are now ready to present the principal result regarding the specific convergence rate achieved by the proposed method, which employs the MMD loss in conjunction with the MWV estimator.
\begin{theorem}\label{theorem:robust_excess_mwv}
	Consider using the MWV method in Eq.~(\ref{eq:m_weighted_var}) as the robust weighting function $\RW$
	with some $\epsilon_h$ slightly larger than $\epsilon$ and $\epsilon_h = O(\epsilon)$.
	Under Assumptions~\ref{ass:weight_identifiability} and \ref{ass:sample_balance}--\ref{ass:kernel}, if the rough estimate $\vq^\prime$ satisfies $\norm{\vq^\prime - \vq^\ast}_2 \leq r_n^\prime$, then with probability at least $1-\tau$, Eq.~(\ref{eq:err_onestep_general}) and Eq.~(\ref{eq:err_multistep_general}) hold
    with the rate $r_n$ specified in Eq.~(\ref{eq:specified_rn}).
\end{theorem}

By integrating Corollary~\ref{corollary:general_multistep_update} with Theorem~\ref{theorem:robust_excess_mwv}, the final estimator resulting from the multi-step refinement achieves a convergence rate of $O(\sqrt{\epsilon / n} + \sqrt{K/(mn)} + \sqrt{K/N})$, up to logarithmic factors. Prior literature \citep{iyer2014maximum, lipton2018detecting, azizzadenesheli2018regularized, garg2020unified} has established theoretical error bounds of $O(\sqrt{{K}/{n}} + \sqrt{{K}/{N}})$ in single-source scenarios without data contamination. Notably, it aligns with our result when $m = 1$ and $\epsilon = 0$.

Similarly, we have the specific convergence rate for the estimator $\vqhat_{\Lcal}$ in Eq.~(\ref{equation:robust_w_divergence}).
\begin{theorem}\label{theorem:robust_divergence_mwv}
	Consider using the MWV method in Eq.~(\ref{eq:m_weighted_var}) as the robust weighting function $\RW$
	with some $\epsilon_h$ slightly larger than $\epsilon$ and $\epsilon_h = O(\epsilon)$.
	Under Assumptions~\ref{ass:sample_balance}, \ref{ass:linearly_independent} and \ref{ass:kernel}(i)--(ii), with probability at least $1 - \tau$, Eq.~(\ref{eq:err_minloss_general}) holds with the rate $r_n$ specified in Eq.~(\ref{eq:specified_rn}).
\end{theorem}

As demonstrated in Lemma~\ref{lemma:upper_var}, the advantage of the robust weighting function based on the excess risk stems from the variance of $\Ecal_{j}(\vq,\vq^\prime)$ being significantly smaller than that of $\Lcal_j(\vq)$ when $\vq$ is close to $\vq^\prime$.
It results in a faster uniform convergence rate of the objective in Eq.~(\ref{equation:robust_w_loss}) in the neighborhood of $\vq^\prime$ \citep{mathieu2021excess}.

\section{Numerical Results}
\label{sec:numerical_study}
In this section, we evaluate the performance of the proposed robust weighted risk minimization approach in Eq.~(\ref{equation:robust_w_loss}) on both synthetic and real-world datasets.
For the construction of the MMD loss, we employ the standard Gaussian kernel, given by $\exp(-{\norm{x_1-x_2}_2^2}/{2})$.
All reported results are averaged over $500$ replications to ensure statistical reliability.

We consider the following methods for the comparison.
\begin{itemize}
	\item \textbf{Single}: Estimates $\vq^\ast$ by minimizing a randomly selected $\Lcal_{j}(\vq)$ where $j \in \Ical$.
	\item \textbf{Average}: Estimates $\vq^\ast$ by minimizing $m^{-1}\sum_{j \in [m]} \Lcal_j(\vq)$, the average of all loss functions.
	\item \textbf{Trim}: Conventional trimmed loss minimization, which estimates $\vq^\ast$ by minimizing $\sum_{j \in [(1-\epsilon_h)m]} \Lcal_{(j)}(\vq)$, where $\Lcal_{(j)}(\vq)$ denotes the $j$-th smallest value in $\{\Lcal_j(\vq): j \in [m]\}$.
	\item \textbf{ROD}: Estimates $\vq^\ast$ by Eq.~(\ref{equation:robust_w_divergence}).
	\item \textbf{ROE}: Estimates $\vq^\ast$ by Eq.~(\ref{equation:robust_w_loss}), with an initial estimate $\vq'$ derived from Eq.~(\ref{equation:robust_w_divergence}).
	\item \textbf{Oracle}: Computes the simple average over the uncontaminated sources.
\end{itemize}

\subsection{Experiments with Synthetic Data}\label{sec:synthetic_data}

We consider a binary classification task with the following data-generating process.
The samples from both the target and source domains are generated independently according to the following specified distributions.
In the target domain, the samples $\{(x_i, y_i)\}_{i \in [N]}$, for which the responses $\{y_i\}_{i \in [N]}$ remain unobserved during training, are drawn independently and identically distributed (i.i.d.) from the joint distribution $(X,Y) \sim Q$, where $\Pr_\target(Y = 1) = 0.6$, $\Pr_\target(Y = 2) = 0.4$, $\target_{X|Y = 1} = \normal(0,1)$ and $\target_{X|Y = 2} = \normal(4,1)$.
Here, $\normal(\mu,\sigma^2)$ denotes the normal distribution with mean $\mu$ and variance $\sigma^2$.
In the $j$-th source domain, the samples $\{(x_{j,i},y_{j,i})\}_{i \in[n]}$ are generated i.i.d. from the distribution $(X,Y) \sim P_j$, where the conditional distributions satisfy
the label shift assumption, such that $P_{j,X|Y = k} = \target_{X|Y = k}$ for $k \in \{1,2\}$ and $\Pr_{P_j}(Y = 1) = p_j$ and $\Pr_{P_j}(Y = 2) = 1-p_j$ with $p_j$ randomly drawn from the uniform distribution over the interval $[0,1]$.
For the contamination, an $\epsilon$ proportion of sources $\{\Dcal_j\}_{j \in [m]}$ is randomly selected as outlier sources $\Ocal$. For each outlier source, a proportion of $5/\sqrt{n}$ of the original labels from the larger category is flipped to the other category.

We use the MWV method as the robust weighting function $\RW$ and assume that $\epsilon$ is known and set $\epsilon_h = \epsilon$ in the algorithms.
It is a promising direction to study data-driven selection of $\epsilon_h$ for the proposed methodology.

For the evaluation metrics, we consider the mean squared error (MSE), $\Ebb \norm{\vqhat - \vq_\ast}_2^2$, of the target class proportion estimator $\vqhat$ and the false selection number of outlier sources (FSN).
The FSN is defined as $\Ebb \abs{\Ocal \cap \Scal}$, where $\Scal$ denotes the non-zero indices of the weights returned from $\RW_{\epsilon_h}(\{\Lcal_j(\vqhat)\})$ or $\RW_{\epsilon_h}(\{\Ecal_j(\vqhat,\vq^\prime)\})$. Both metrics reflect better performance when their values are lower.

For the parameters $(N, m, n, \epsilon)$, we adopt the following configuration.
The sample size of the target domain is set at $N=nm$.
One of the rest $(m, n, \epsilon)$ is varied, while the other two are maintained at their baseline values of $(m,n,\epsilon)=(40,100,0.2)$.

The results are presented in Figure~\ref{figure:plotrdvsrer} across varying parameter values for $(m,n,\epsilon)$.
The Average approach demonstrates suboptimal performance, primarily attributable to substantial bias induced by data corruption.
The Single approach is inadequate, as it fails to fully exploit the information available from all sources.
Conversely, the Trim, ROD, and ROE methodologies exhibit enhanced accuracy in aggregating information from multiple sources and display robustness in detecting outlier sources.
Notably, the proposed ROE method yields a lower MSE than the Trim and ROD methods.
Furthermore, ROE has a smaller FSN compared to both Trim and ROD, thereby enhancing its robustness against contamination.
Remarkably, even when the contamination proportion increases to $40\%$, the ROE method continues to perform relatively well, maintaining a performance level comparable to that of the Oracle.
This superiority is consistent with expectations regarding the improved outlier identification capacity of the ROE method, as illustrated in the second row of Figure~\ref{figure:plotrdvsrer}.
These empirical findings align with the theoretical insights presented in Section~\ref{sec:theory}.

\begin{figure}[t]
	\centering
	\includegraphics[width=0.96\linewidth]{"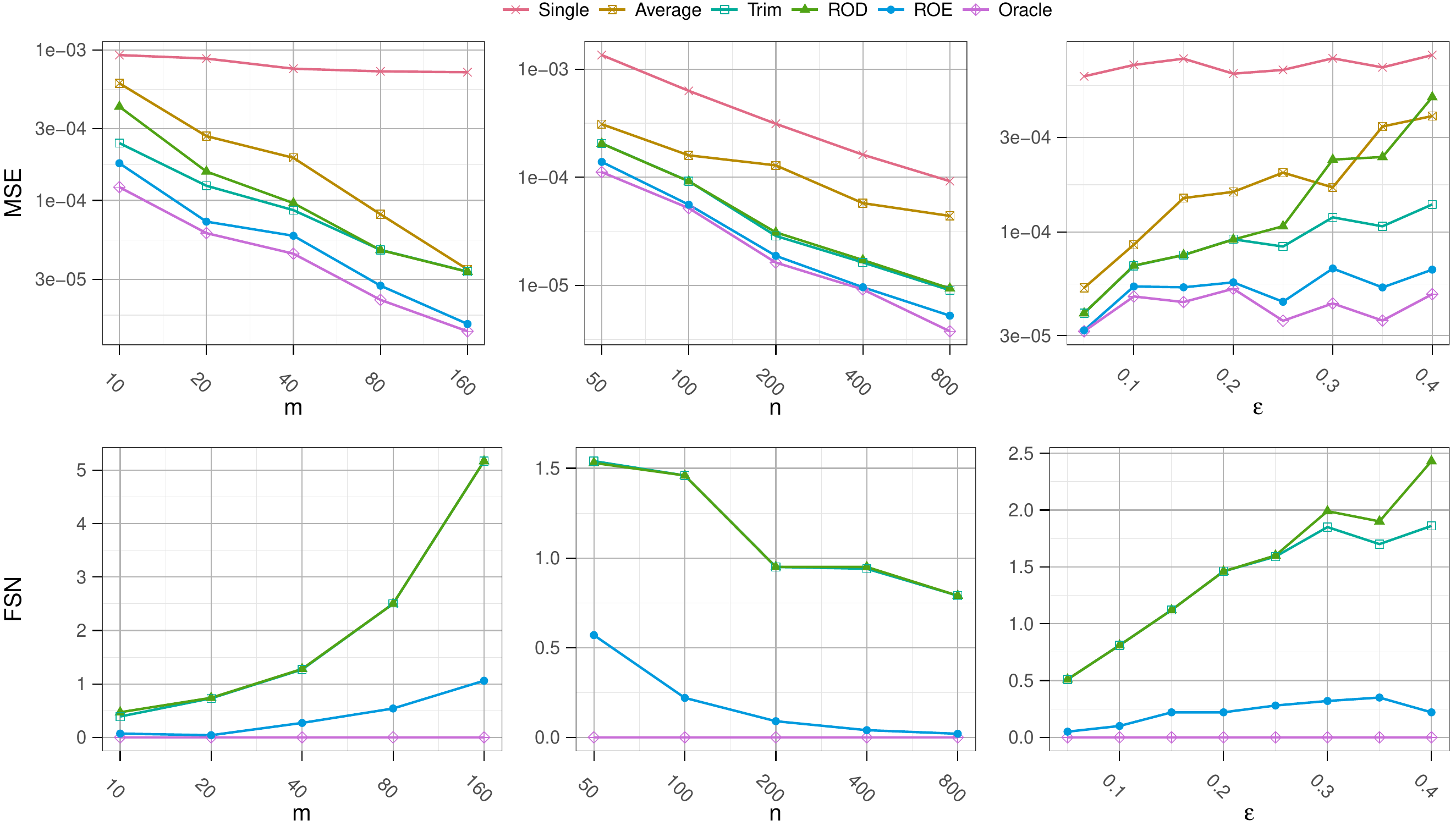"}
	\caption{MSE of $\vqhat$ and FSN for the comparison methods in the synthetic task.}
	\label{figure:plotrdvsrer}
\end{figure}

\subsection{Experiments with Real Data}
We further demonstrate the efficacy of the proposed ROE method through two real-world case studies.
Besides the MWV estimator, we also explore the truncated mean as an alternative robust weighting function $\RW$, which eliminates $\epsilon$ proportion of the extreme values from both ends of the loss values. 
For comprehensive details, please refer to the Appendix. 
The methodologies that employ the truncated mean estimator are denoted as ROE-TRU and ROD-TRU, while those using the MWV estimator are denoted as ROE-MWV and ROD-MWV, respectively.

Besides analyzing the estimation error of $\hat{\vq}$, we also investigate the impact of these methods on downstream classification tasks after obtaining $\hat{\vq}$, as described in Remark~\ref{remark:robust_classifier}. We present the results in terms of the misclassification error of the learned classifiers on the target domain samples.

\subsubsection{Remote Sensing Application}
The Zurich-Summer dataset is designed for the classification of super-pixels derived from high-resolution satellite imagery into various land cover classes \citep{volpi2015semantic}.
The dataset consists of $20$ large-scale satellite images, each containing four spectral channels.

\begin{table}[t]
    \centering
    \setlength\tabcolsep{0.1em}
    \renewcommand{\arraystretch}{1}
    \begin{threeparttable}
    {\scriptsize 
	\begin{tabular*}{1\textwidth}{@{\extracolsep{\fill}}c*{16}{r}}
		\toprule
		\multirow{2}{*}{Method} & \multicolumn{15}{c}{Image index of the target domain} & \multirow{2}{*}{Average} \\
		& 1 & 2 & 3 & 4 & 5 & 6 & 7 & 8 & 9 & 10 & 11 & 12 & 13 & 14 & 15 & \\
		\midrule
		& \multicolumn{15}{c}{MSE ($ \times $100)} & \\
		Single & 9.27 & 10.05 & 7.68 & 6.30 & 8.35 & 4.83 & 6.67 & 11.24 & 15.94 & 14.75 & 7.33 & 9.93 & 15.35 & 9.63 & 7.00 & 9.62 \\
		Average & 8.43 & 9.69 & 6.03 & 5.73 & 5.20 & 8.68 & 5.08 & 9.16 & 12.76 & 6.99 & 8.64 & 6.49 & 8.78 & 5.79 & 7.00 & 7.63 \\
		Trim & 5.54 & 8.62 & 2.60 & 4.13 & 5.46 & \pmb{3.87} & 2.02 & 9.82 & 16.61 & 8.47 & \pmb{2.37} & 5.74 & 12.68 & 4.89 & \pmb{2.08} & 6.33 \\
		ROD-MWV & 5.46 & 7.94 & 2.87 & 4.06 & 5.22 & 4.34 & 2.32 & 9.16 & 14.95 & 8.13 & 2.70 & 5.37 & 11.54 & 5.06 & 2.73 & 6.12 \\
		ROE-MWV & 4.79 & \pmb{6.93} & 2.43 & 3.58 & 4.67 & 4.49 & 2.68 & \pmb{6.84} & \pmb{9.14} & 6.48 & 2.69 & \pmb{4.98} & \pmb{8.26} & \pmb{4.54} & 2.96 & \pmb{5.03} \\
		ROD-TRU & 6.99 & 8.86 & 3.52 & 4.48 & 5.65 & 4.56 & 2.75 & 10.46 & 17.74 & 9.30 & 3.34 & 6.27 & 13.12 & 5.99 & 2.52 & 7.04 \\
		ROE-TRU & \pmb{4.33} & 7.10 & \pmb{2.02} & \pmb{3.09} & \pmb{4.47} & 4.33 & \pmb{1.90} & 7.65 & 13.14 & \pmb{6.36} & 2.43 & 5.30 & 9.79 & 4.84 & 2.64 & 5.29 \\
		\midrule
		& \multicolumn{15}{c}{Misclassification error(\%)} & \\
		Single & 32.94 & 30.14 & 27.29 & 26.70 & 29.18 & 25.13 & 29.39 & 29.93 & 32.81 & 35.86 & 28.18 & 40.08 & 32.68 & 35.09 & 29.73 & 31.01 \\
		Average & 41.47 & 36.59 & 30.65 & 33.29 & 31.74 & 36.04 & 32.95 & 34.07 & 36.31 & 34.58 & 33.18 & 44.00 & 33.69 & 30.92 & 34.36 & 34.92 \\
		Trim & 30.05 & 31.36 & 19.85 & 21.59 & 27.53 & 23.46 & 22.41 & 30.88 & 34.60 & 32.60 & 21.40 & 39.40 & 30.89 & 29.01 & \pmb{23.22} & 27.88 \\
		ROD-MWV & 29.77 & 32.28 & 19.75 & 21.75 & 28.24 & 25.14 & 23.37 & 31.33 & 35.83 & 33.65 & 22.49 & 39.74 & 31.57 & 28.89 & 23.85 & 28.51 \\
		ROE-MWV & \pmb{27.21} & \pmb{29.97} & 17.55 & \pmb{18.03} & \pmb{27.40} & \pmb{21.57} & 22.44 & \pmb{29.22} & \pmb{32.17} & \pmb{31.42} & 20.24 & \pmb{39.03} & \pmb{30.36} & \pmb{27.62} & 23.45 & \pmb{26.51} \\
		ROD-TRU & 32.25 & 31.68 & 22.41 & 22.48 & 28.45 & 25.67 & 24.30 & 31.19 & 34.92 & 32.77 & 24.22 & 39.63 & 31.26 & 31.24 & 24.33 & 29.12 \\
		ROE-TRU & 28.76 & 31.71 & \pmb{17.41} & 20.25 & 28.25 & 22.55 & \pmb{21.84} & 30.58 & 34.04 & 31.80 & \pmb{20.20} & 39.59 & 31.54 & 28.13 & 23.24 & 27.33 \\
		\bottomrule
        \end{tabular*}%
    }
    \end{threeparttable}
    \caption{MSE of $\vqhat$ and misclassification error for the comparison methods on the Zurich-Summer dataset.
        Results with the best performance are highlighted in bold.}
    \label{table:zurich}
\end{table}

For the study, we selected pixels corresponding to roads, buildings, trees, and grass, ensuring that these classes are represented across all images. Due to the inherent class imbalance across different images, there is no need to artificially alter the label proportions to introduce label shift.
In each experimental iteration, one image is designated as the target domain, from which $N=20,000$ pixels are randomly extracted to serve as target samples.
The remaining $19$ images are used as source domains, with $n=1,000$ pixels randomly extracted from each image to constitute the source samples.
For the contamination, an $\epsilon = 30\%$ proportion of sources are randomly selected as the outlier sources $\Ocal$. For each outlier source, all the original labels undergo the following cyclic perturbation: road $\rightarrow$ building $\rightarrow$  tree $\rightarrow$ grass $\rightarrow$ road.

The results are summarized in Table~\ref{table:zurich}, which presents the MSE of $\hat{\vq}$ and the misclassification error for scenarios in which the first $15$ images are successively selected as the target domain. 
The proposed ROE-based methods demonstrate a superior performance relative to other methodologies across the majority of these $15$ cases, achieving an overall enhancement in performance, thereby underscoring its effectiveness and robustness.

\subsubsection{Digital Recognition}
We explore the digital recognition task using the MNIST dataset, which consists of 60,000 training images and 10,000 test images, with labels ranging from digital``0" to digital ``9".
An autoencoder is employed and trained on the training images to extract four features from each image.
In this experimental setup, the test images are designated as target samples with balanced label proportion $(0.1,\ldots,0.1) \in \triangle_{10}$, while the training images are randomly allocated to $m=40$ sources. To introduce label shift between target and source domains, each source consists of $300$ images randomly drawn from the training dataset, with label proportion given by $(0.12,0.08,0.12,0.08,\dots,0.12,0.08) \in \triangle_{10}$.
For the contamination, an $\epsilon$ proportion of sources are randomly selected as the outlier sources $\Ocal$. For each outlier source, $50\%$ of the clean labels in classes 0 to 3 are shifted one position to the right, as illustrated below: $ 0\rightarrow1\rightarrow2\rightarrow 3 \rightarrow 4$. Various values of the corruption proportion $\epsilon$ are investigated, ranging from $5\%$ to $40\%$.

Table~\ref{table:mnist} summarizes the results in terms of the MSE of $\vqhat$ and the misclassification error. The results highlight that the ROE method outperforms other approaches, yielding lower MSE and misclassification rate, thereby underscoring the robustness and efficacy in the task. Methods based on truncation perform slightly worse than those based on MWV, as they remove a larger proportion of inlier sources, leading to a loss in efficiency.

\begin{table}[t]
	\centering
    \setlength\tabcolsep{0.3em}
    \renewcommand{\arraystretch}{1}
    \begin{threeparttable}
		\begin{tabular*}{.9\linewidth}{@{\extracolsep{\fill}}c*{8}{r}}
			\toprule
			\multirow{2}{*}{Method} & \multicolumn{8}{c}{Corruption proportion $ \epsilon$} \\
			& 0.05 & 0.10 & 0.15 & 0.20 & 0.25 & 0.30 & 0.35 & 0.40 \\
			\midrule
			 & \multicolumn{8}{c}{MSE ($1,000^{-1}$)} \\
			Single & 2.04 & 1.88 & 1.91 & 2.21 & 1.83 & 1.98 & 1.87 & 1.73 \\
                Average & 0.36 & 0.47 & 0.60 & 0.78 & 1.00 & 1.28 & 1.58 & 1.93 \\
                Trim & 0.33 & 0.36 & 0.35 & 0.35 & 0.34 & 0.36 & 0.34 & 0.36 \\
                ROD-MWV & 0.33 & 0.36 & 0.35 & 0.35 & 0.34 & 0.36 & 0.36 & 0.35 \\
                ROE-MWV & $\pmb{0.29}$ & $\pmb{0.26}$ & $\pmb{0.25}$ & $\pmb{0.26}$ & $\pmb{0.25}$ & $\pmb{0.28}$ & $\pmb{0.28}$ & $\pmb{0.31}$ \\
                ROD-TRU & 0.34 & 0.39 & 0.40 & 0.44 & 0.45 & 0.53 & 0.57 & 0.77 \\
                ROE-TRU & $\pmb{0.29}$ & 0.28 & 0.28 & 0.27 & 0.28 & 0.35 & 0.37 & 0.54 \\
			\midrule
			 & \multicolumn{8}{c}{Misclassification error (\%)} \\
			Single & 14.14 & 14.15 & 14.09 & 13.94 & 14.05 & 13.99 & 14.02 & 13.90 \\
                Average & 11.38 & 11.43 & 11.64 & 12.02 & 12.63 & 13.51 & 14.71 & 16.18 \\
                Trim & 11.47 & 11.52 & 11.52 & 11.52 & 11.50 & 11.53 & 11.52 & 11.53 \\
                ROD-MWV & 11.47 & 11.52 & 11.51 & 11.52 & 11.49 & 11.52 & 11.51 & 11.49 \\
                ROE-MWV & $\pmb{10.29}$ & $\pmb{10.30}$ & $\pmb{10.28}$ & $\pmb{10.29}$ & $\pmb{10.25}$ & $\pmb{10.27}$ & $\pmb{10.29}$ & $\pmb{10.27}$ \\
                ROD-TRU & 11.47 & 11.53 & 11.54 & 11.57 & 11.55 & 11.62 & 11.68 & 11.76 \\
                ROE-TRU & 10.31 & 10.32 & 10.32 & 10.33 & 10.34 & 10.39 & 10.44 & 10.54 \\
			\bottomrule
		\end{tabular*}%
    \end{threeparttable}
    \caption{MSE of $\vqhat$ and misclassification error for the comparison methods on the MNIST dataset. Results with the best performance are highlighted in bold.}
    \label{table:mnist}
\end{table}

\section{Concluding Remarks}\label{sec:conclusion}

In this paper, we present a robust solution to the unsupervised label shift adaption problem involving multiple sources, some of which may not conform to the label shift assumption.
We propose a robust estimation method for the target class proportion employing a weighted empirical risk minimization framework. 
The weights are determined through robust weighting algorithms on the excess risk values, which facilitates the detection of outlier sources through the variance reduction property.
Theoretically, we developed a general convergence theory for the proposed approach with general divergence and robust weighting methods.
By using MMD as the measure of distribution divergence and MWV as the robust weighting method, we provide specific theoretical error guarantees that match the optimal results in the robust learning literature. Extensive numerical experiments are also conducted to validate the proposed method's robustness and efficiency, demonstrating its superior performance.

Nonetheless, several advancing research issues are imperative. First, while a larger $\epsilon_h$ may enhance the model's robustness to contamination, excessively high values could undermine performance, illustrating a crucial trade-off. Therefore, data-driven selection of the hyperparameter $\epsilon_h$ deserves further studies. Second, we assume that a proportion of sources are corrupted in this work. It is promising to study the robust estimation task that all data sources are contaminated with a certain proportion of corrupted samples. 
Third, our theoretical analysis predominantly focuses on the setting with bounded variance of the loss function, such as the MMD loss with bounded kernel assumption. Relaxing the bounded variance condition, e.g., to the infinite variance setting \citep{xu2023infinite} and extending the results to other distribution measures, such as the measures based on optimal transport \citep{redko2019optimal}, require further studies. These considerations will be the subject of future research endeavors.

\appendix

\begin{center}
    \LARGE{Appendix}
\end{center}

\noindent The Appendix includes the proof of the theoretical results in Section~\ref{sec:theory}, further discussions regarding the existing literature on label shift adaptation and robust estimation, and additional numerical results.

\section{Proof of Theoretical Results}
This section provides detailed proofs of the theoretical results stated in Section~\ref{sec:theory}.
The section is organized as follows.
In Sections~\ref{sec:general_thm}--\ref{sec:proof_general_robust_loss}, we provide proofs for the theoretical results related to general distribution divergence losses and robust weighting functions.
Subsequently, in Section~\ref{sec:lem_mwv}, we explore the theoretical properties of the MWV estimator introduced in Eq.~(\ref{eq:m_weighted_var}).
With these properties, in Sections~\ref{sec:proof_upper_var}--\ref{sec:robust_divergence_mwv}, we continue to prove the specific results when the MMD loss and the MWV estimator are adopted. At the end, we illustrate the variance reduction phenomenon with the linear model example.

For the sake of simplicity, we establish the commonly used notation for the proof.
Let $\triangle_{m}$ be the $ m $-dimensional probability simplex.
Denote $\triangle_{m,\epsilon} = \{\w \in \triangle_{m}: w_{j} \leq m^{-1}(1-{\epsilon})^{-1} \}$, which indicates that at most ${\epsilon} m$ of the entries in $\w \in \triangle_{m,\epsilon}$ are zero.
Let $\Scal = \{x_i\}_{i \in [m]} \subset \real$ denotes a set of $m$ points.
Let $\mu_\Scal = \abs{\Scal}^{-1} \sum_{x_i \in \Scal} x_i$ and $ \sigma^2_{\Scal,\mu} = \abs{\Scal}^{-1} \sum_{x_i \in \Scal} (x_i - \mu)^2 $.
For a weight $\w \in \triangle_{m}$, let $\mu_{\w,\Scal} = \sum_{x_i \in \Scal} w_i x_i$ and $ \sigma^2_{\w,\Scal,\mu} = \sum_{x_i \in \Scal} w_i(x_i - \mu)^2 $.
Denote $\sigma^2_{\Scal} =\sigma^2_{\Scal,\mu_{\Scal}}$ for simplicity.
The weighted average estimator defined in Assumption~\ref{ass:uniform_converge_excess_risk} can be formulated as,
\[\RE_{\epsilon_h}(\{z_j\}_{j \in [m]}) = \sum_{j \in [m]} w_{j} z_j \text{ s.t. } \w = \RW_{\epsilon_h}(\{z_j\}_{j \in [m]}),\]
where $\{z_j\}_{j \in [m]} \subset \real$.

\subsection{Proof of Theorem~\ref{theorem:general_one_step_update}}\label{sec:general_thm}
\begin{proof}
	We specify the local minimizer to which the theorem refers.
	As demonstrated in Section~\ref{sec:optimization}, $\w = \RW_{\epsilon_h}(\{\Ecal_{j}(\vq,\vq')\})$ can be viewed as a function of $\vq$, i.e., $\w = \w(\vq)$.
	By Assumption~\ref{ass:weight_identifiability}, in the region $\{\vq: \lVert \vq - \vq^\ast\rVert_2 < C \sqrt{r_n}\}$,
	the optimization problem (\ref{equation:robust_w_loss}) share the same local minimizers with the problem,
	\begin{equation}\label{equation:robust_excess_risk}
		\vqhat_\Ecal = \argmin_{\vq \in \triangle_{K}} \sum_{j \in [m]} w_{j} \Ecal_{j}(\vq,\vq^\prime) \text{ s.t. } \w = \RW_{\epsilon_h}(\{\Ecal_{j}(\vq,\vq^\prime)\}).
	\end{equation}
	We will show that the global minimizer $\vqhat_\Ecal$ of Eq.~(\ref{equation:robust_excess_risk}) satisfies $\norm{\vqhat_\Ecal - \vq^\ast}_{2} \lesssim \max\{r_n,\sqrt{r_n r_n^\prime}\}$.
	Then since $r_n' \le 1$ and $r_n = o(1)$, we have $\max\{r_n,\sqrt{r_n r_n^\prime}\} \le \sqrt{r_n}$ and $\vqhat_\Ecal \in \{\vq: \lVert \vq - \vq^\ast\rVert_2 \le C \sqrt{r_n}\}$.
	Therefore, $\vqhat_\Ecal$ is the local minimizer with the desired error bound in the theorem.

	The error bound of $\vqhat_\Ecal$ comes from the following fact.
	\begin{equation*}
		\begin{aligned}
		& \lambda \norm{\vqhat_\Ecal - \vq^\ast}_{2}^{2}
        \leq \Ecal_\target(\vqhat_\Ecal, \vq^\ast) = \Ecal_\target(\vqhat_\Ecal, \vq^\prime) - \Ecal_\target(\vq^\ast,\vq^\prime) \\
		\leq &\Ecal_\target(\vqhat_\Ecal, \vq^\prime) - \Ecal_\target(\vq^\ast,\vq^\prime) - \hat{\Ecal}(\vqhat_\Ecal, \vq^\prime) + \hat{\Ecal}(\vq^\ast,\vq^\prime) \\
		\leq& 2 \max(\norm{\vqhat_\Ecal - \vq^\ast}_{2} + \norm{\vq^\prime - \vq^\ast}_{2}, r_n) \sup_{\vq_1, \vq_2 \in \triangle_{K}} \frac{\abs{\hat{\Ecal}(\vq_1,\vq_2) - \Ecal_Q(\vq_1,\vq_2)}}{\max(\norm{\vq_1 - \vq^\ast}_{2} + \norm{\vq_2 - \vq^\ast}_{2}, r_n)} \\
		\lesssim& \max(\norm{\vqhat_\Ecal - \vq^\ast}_{2}, r_n^\prime,r_n) \times r_n,
		\end{aligned}
	\end{equation*}
	where the first inequality is due to Assumption~\ref{ass:convex}, the second inequality comes from the fact that $\vqhat_\Ecal$ is the minimizer of $\hat{\Ecal}(\vq,\vq^\prime)$, and the last inequality holds with probability at least $1 - \tau$, which is due to Assumption~\ref{ass:uniform_converge_excess_risk}.
	After reorganization, we obtain the desired result:
	\[
        \norm{\vqhat_\Ecal - \vq^\ast}_{2} \lesssim \max\{r_n, \sqrt{r_n r_n^\prime}\}.
    \]
\end{proof}

\subsection{Proof of Corollary~\ref{corollary:general_multistep_update}}\label{sec:proof_cor}
\begin{proof}
    From the proof of Theorem~\ref{theorem:general_one_step_update}, we know that the probability $1 - \tau$ is from Assumption~\ref{ass:uniform_converge_excess_risk} and holds uniformly for all initial estimates $\vq^\prime$.
	By iteratively applying Theorem~\ref{theorem:general_one_step_update}, the results can be obtained.
\end{proof}

\subsection{Proof of Theorem~\ref{theorem:general_robust_loss}}\label{sec:proof_general_robust_loss}
\begin{proof}
	Note that
	\begin{equation*}
		\begin{aligned}
			\lambda \norm{\vqhat_\Lcal - \vq^\ast}_{2}^{2} & \leq \Lcal_\target(\vqhat_\Lcal) - \Lcal_\target(\vq^\ast) \\
			 & \leq \Lcal_\target(\vqhat_\Lcal) - \hat{\Lcal}(\vqhat_\Lcal)
			+ \hat{\Lcal}(\vq^\ast) - \Lcal_\target(\vq^\ast) \\
			 & \leq 2 \sup_{\vq \in \triangle_{K}} |\Lcal_\target(\vq) - \hat{\Lcal}(\vq)| \\
			 & \lesssim r_n,
		\end{aligned}
	\end{equation*}
	where the first inequality is due to Assumption~\ref{ass:convex}, the second inequality is from the fact that $\vqhat_\Lcal$ is the minimizer of $\hat{\Lcal}(\vq)$, and the last inequality is due to Assumption~\ref{ass:uniform_converge_divergence}.
	The desired conclusion follows immediately by taking the square root of both sides of the above inequality.
\end{proof}

\subsection{Properties of the MWV Estimator}\label{sec:lem_mwv}
In this section, we study the theoretical properties of the MWV estimator
\[
	\mathrm{MWV}_{\epsilon_h} (\{z_j\}_{j\in [m]}) = \abs{\widehat{\Jcal}}^{-1} \sum_{j \in \widehat{\Jcal}} z_j,
\]
where $\widehat{\Jcal}$ is defined in Eq.~(\ref{eq:m_weighted_var}).
The goal of the section is to prove that the MWV estimator satisfies the requirements in Assumptions~\ref{ass:uniform_converge_excess_risk} and \ref{ass:uniform_converge_divergence}, which are provided in Lemma~\ref{lemma:sup_of_adaptive_RB}.

Before proving Lemma~\ref{lemma:sup_of_adaptive_RB}, we first introduce the definition of $(\epsilon,\delta)$-stability and list several properties about it with self-contained proofs.
We refer to \cite{diakonikolas2020outlier} for more details.
\begin{definition}[$ (\epsilon, \delta) $-stability]
	Let $\delta>0$, $0<\epsilon<1/2$, we say that $\Scal$ is $(\epsilon, \delta)$-stable with respect to $\mu$ if, for any subset $\Scal^\prime \subseteq \Scal$ with $|\Scal^\prime| \geq m(1 - \epsilon)$, the following conditions hold: 1) $ |\mu_{\Scal^\prime} - \mu| \leq \delta $ and 2) $ \sigma^2_{\Scal^\prime,\mu} \leq \delta^2 / \epsilon$.
\end{definition}

The following lemma demonstrates that after removing/adding some points in an $(\epsilon,\delta)$-stable set $\Scal$, the error of the mean of the new set $\Scal^\prime$ can be controlled once the sample variance of $\Scal^\prime$ is bounded.
\begin{lemma}\label{lemma:efficiency_for_methods_based_on_stability}
	If $\Scal$ is an $(\epsilon, \delta)$-stable set with respect to $\mu$, and for any set $\Scal_1$ such that $|\Scal \cap \Scal_1| = (1 - \epsilon_1) |\Scal_1|$, $|\Scal \cap \Scal_1| \geq (1 - \epsilon) |\Scal|$ and $\sigma^2_{\Scal_1} \leq \xi$, then we have $ |\mu_{\Scal_1} - \mu| \leq \delta + \sqrt{ \xi \epsilon_1/(1-\epsilon_1)}$.
\end{lemma}
\begin{proof}
	Let $\Tcal = \Scal \cap \Scal_1$ and $\Tcal_1 = \Scal_1 \setminus \Scal$.
	A straightforward calculation gives
	\[\sigma^2_{\Scal_1}= (1-\epsilon_1)\sigma^2_{\Tcal} + \epsilon_1 \sigma^2_{\Tcal_1} + \epsilon_1 (1-\epsilon_1) (\mu_{\Tcal} - \mu_{\Tcal_1})^2,\]
	which implies
	\[|\mu_{\Tcal} - \mu_{\Tcal_1}| \leq \sqrt{\frac{\sigma^2_{\Scal_1}}{\epsilon_1(1-\epsilon_1)}} \leq \sqrt{ \frac{\xi}{ \epsilon_1 (1-\epsilon_1)}} .\]
	The final result is obtained through the following calculation:
	\begin{align*}
		|\mu_{\Scal_1} - \mu| & = |(1-\epsilon_1)\mu_{\Tcal} + \epsilon_1 \mu_{\Tcal_1} - \mu| \\
		 & \leq |\mu_{\Tcal} - \mu | + \epsilon_1 |\mu_{\Tcal} - \mu_{\Tcal_1}| \\
		 & \leq \delta + \sqrt{ \xi \epsilon_1/(1-\epsilon_1)},
	\end{align*}
	where we use the $ (\epsilon,\delta) $-stability of $\Scal$ and $\abs{\Tcal} \geq (1-\epsilon)\abs{\Scal}$.
\end{proof}

The following lemma provides sufficient deterministic conditions for a set to be $(\epsilon,\delta)$-stable.
\begin{lemma}\label{lemma:entire_set_stability_to_subset_stability}
	If $|\mu_\Scal - \mu| \leq \delta_1$ and $\sigma^2_{\Scal,\mu} \leq \delta_1^2/\epsilon_1$, then for any $\epsilon_2 \leq 1/2$, $\Scal$ is $(\epsilon_2, \delta_2)$-stable with respect to $\mu$, where $\delta_2 = \delta_1(1+\sqrt{2\epsilon_2/\epsilon_1})$.
\end{lemma}
\begin{proof}
	For any subset $\Scal_1 \subset \Scal$ and $|\Scal_1|\geq (1-\epsilon_2)m$, the variance of $\Scal_1$ is bounded by
	\begin{align*}
		\sigma^2_{\Scal_1,\mu} \leq \frac{1}{1-\epsilon_2} \sigma^2_{\Scal,\mu} \leq 2\delta_1^2/\epsilon_1 \leq \delta_2^2/\epsilon_2.
	\end{align*}
	Denote $\w = (w_i)_{i\in [m]}$ such that $w_i = 1/\abs{\Scal_1}$ if $x_i \in \Scal_1$ otherwise $0$.
	The mean of $\Scal_1$ is bounded by
	\begin{align*}
		|\mu_{\Scal_1} - \mu| & \leq |\mu_{\Scal_1} - \mu_\Scal| + |\mu_\Scal - \mu| \\
		 & = \Bigabs{\sum_{i \in [m]} (w_i - 1/m) (x_i - \mu)} + |\mu_\Scal - \mu| \\
		 & \leq \sqrt{m \sum_{i \in [m]} (w_i - 1/m)^2} \sqrt{m^{-1}\sum_{i \in [m]} (x_i - \mu)^2} + \delta_1 \\
		 & \leq \sqrt{2\epsilon_2} \times \sqrt{\delta_1^2/\epsilon_1} + \delta_1
		\leq \delta_2,
	\end{align*}
	which follows from Hölder's inequality.
\end{proof}
The following lemma shows that if $\Scal$ contains a stable subset with respect to $\mu$, then the MWV method produces a robust estimator of $\mu$ with a small error.

\begin{lemma}\label{lemma:stability_on_corrupted_set}
	Assume that there exists a subset $\Scal_1 \subset \Scal$ with $\abs{\Scal_1} = (1-\epsilon_1) \abs{\Scal}$ such that $\abs{\mu_{\Scal_1} - \mu} \leq \delta_1$ and $\sigma^2_{\Scal_1,\mu} \leq \delta_1^2 /\epsilon_1$.
	Choose $\epsilon_h$ as $ \epsilon_1 \le \epsilon_h <1/2$.
	Let the MWV mean estimator $\mu_{\Scal^\ast}$ be determined by solving
	$$\Scal^\ast = \argmin_{\Jcal \subset \Scal, |\Jcal| \geq (1 - \epsilon_h) m} \sigma_{\Jcal}^2.$$
	We have $\abs{\mu_{\Scal^\ast} - \mu} \leq O(\delta_1 (1 + \sqrt{\epsilon_h / \epsilon_1}))$.
	Additionally if $\epsilon_h = O(\epsilon_1)$, then $\abs{\mu_{\Scal^\ast} - \mu} \leq O(\delta_1)$.
\end{lemma}
\begin{proof}
	Given that $|\mu_{\Scal_1} - \mu| \leq \delta_1$ and $\sigma^2_{\Scal_1,\mu} \leq \delta_1^2/\epsilon_1$, by Lemma~\ref{lemma:entire_set_stability_to_subset_stability}, $\Scal_1$ is $(\epsilon_2, \delta_2)$-stable with respect to $\mu$, where $\epsilon_2 = \epsilon_h/(1-\epsilon_1) $ and $\delta_2 = \delta_1 (1+\sqrt{2\epsilon_2/\epsilon_1})$.
	Since $\abs{\Scal_1} = (1-\epsilon_1) \abs{\Scal} \geq (1-\epsilon_h) \abs{\Scal}$, we have $ \sigma^2_{\Scal^\ast,\mu} \leq \sigma^2_{\Scal_1,\mu} \leq \delta_1^2/\epsilon_1$ by the definition of $\Scal^\ast$.

	Since $\abs{\Scal^\ast \cap \Scal_1} \geq (1-\epsilon_h-\epsilon_1)m $, we have
	\[
		\abs{\Scal^\ast \cap \Scal_1} \geq \frac{1-\epsilon_h-\epsilon_1}{1-\epsilon_h} \abs{\Scal^\ast} = (1-\frac{\epsilon_1}{1-\epsilon_h}) \abs{\Scal^\ast},
	\]
	and
	\[
		\abs{\Scal^\ast \cap \Scal_1} \geq \frac{1-\epsilon_h-\epsilon_1}{1-\epsilon_1} \abs{\Scal_1} = (1-\epsilon_2) \abs{\Scal_1}.
	\]
	Applying Lemma~\ref{lemma:efficiency_for_methods_based_on_stability} with $\xi = \delta_1^2/\epsilon_1$, we obtain
	\[
		|\mu_{\Scal^\ast} - \mu| = \delta_2 + \sqrt{\frac{\delta_1^2}{\epsilon_1} \frac{\epsilon_1/(1 - \epsilon_h)}{1- \epsilon_1/(1 - \epsilon_h)}} = O\biggl(\delta_1 \biggl(1 + \sqrt{\frac{\epsilon_h}{ \epsilon_1}}\biggr)\biggr).
	\]
	And when $\epsilon_h = O(\epsilon_1)$, $|\mu_{\Scal^\ast} - \mu| = O(\delta_1)$.
\end{proof}

By applying Lemma~\ref{lemma:entire_set_stability_to_subset_stability}, we show that when the distributions of the samples are with bounded variance, then a stable subset exists with high probability.
\begin{lemma}\label{lemma:control_the_variance}
	Let $\Scal = \{x_i\}_{i \in [m]}$ be a set of $m$ random variables such that $\E[x_i] = \mu$ and $\max_{i \in [m]} \Var(x_i) \leq \sigma^2 < \infty$.
	For $\epsilon^\prime \ge C\log(1/\tau) / m$ with the constant $C > 0$ specified in the proof, with probability at least $1 - \tau$, there exists a subset $\Scal_1 \subset \Scal$ with $\abs{\Scal_1} \geq (1-\epsilon^\prime) \Scal$ such that
	\[
		|\mu_{\Scal_1} - \mu| \leq 6 \sigma \sqrt{\epsilon^\prime}, \, \sigma^2_{\Scal_1,\mu} \leq \frac{2 \sigma^2}{1-\epsilon^\prime}.
	\]
\end{lemma}
\begin{proof}
	Without loss of generality, assume that $\mu = 0$.
	We can choose $\Scal_1 = \Scal_2 \cap \Scal$ with $\Scal_2 = \{z_i: z_i = x_i \indicator{\{\abs{x_i}\leq 2\sigma/\sqrt{\epsilon^\prime}\}} \}$.
	We will show that $\Scal_1$ satisfies the three requirements listed in the Lemma.

	Firstly, we establish a lower bound for the size of $\Scal_1$.
	Applying Chebyshev's inequality, we obtain $ \Pr (|x_i| > 2\sigma/\sqrt{\epsilon^\prime}) \leq \E[\mathbbm{1}{\{|x_i| > 2 \sigma/\sqrt{\epsilon^\prime}\}}] \leq \epsilon^\prime /4 $ and $ \Var[\mathbbm{1}{\{|x_i| > 2 \sigma/\sqrt{\epsilon^\prime}\}}] \leq \epsilon^\prime/4$.
	Then applying Bennett’s inequality \citep[see][Theorem 2.9.2]{vershynin2020high}, we have
	\begin{align*}
		\Pr\Bigl(\sum_{x_i \in \Scal}\bigl\{\indicator{\{|x_i| > 2 \sigma/\sqrt{\epsilon^\prime}\}} - \E[\mathbbm{1}{\{|x_i|
				> 2 \sigma/\sqrt{\epsilon^\prime}\}}]\bigr\} > \frac{1}{2} \epsilon^\prime m \Bigr) \leq \exp(- c_1 \epsilon^\prime m )
	\end{align*}
	for some constant $c_1$, which indicates
	\begin{equation}\label{eq:num_of_removed_points}
		\begin{aligned}
			\abs{\Scal_1} & = m - \sum_{x_i \in \Scal} \mathbbm{1}{\{|x_i| > 2 \sigma/\sqrt{\epsilon^\prime}\}} \\
			 & \geq m - \sum_{x_i \in \Scal} \E[\mathbbm{1}{\{|x_i| > 2 \sigma/\sqrt{\epsilon^\prime}\}}] - \frac{1}{2} \epsilon^\prime m \\
			 & \geq (1 - \epsilon^\prime) m
		\end{aligned}
	\end{equation}
	with probability at least $1-\exp(- c_1 \epsilon^\prime m )$.

	Next, we establish the bound for $\sigma^2_{\Scal_1,\mu}$.
	Note that
	\begin{equation*}
		\sigma^2_{\Scal_1,\mu} = \frac{1}{\abs{\Scal_1}} \sum_{z_i \in \Scal_1} z_i^2 = \frac{1}{\abs{\Scal_1}} \sum_{i \in [m]} z_i^2.
	\end{equation*}
	It is sufficient to derive an upper bound for $\sum_{i \in \Scal_2} z_i^2$.
	Given that $\E[z_i^2] \leq \E[x_i^2] \leq \sigma^2$, $\Var(z_i^2) \leq \E[z_i^4] \leq 4 \E[x_i^2 \sigma^2/\epsilon^\prime] \leq 4 \sigma^4/\epsilon^\prime$ and $\abs{z_i} \leq 2\sigma /\sqrt{\epsilon^\prime}$, by Bennett's inequality \citep[see][Theorem 2.9.2]{vershynin2020high}, we have
	\[\Pr\Bigl( \sum_{i \in [m]} (z_i^2 - \E[z_i^2]) > m\sigma^2 \Bigr) \leq \exp ( - c_2 \epsilon^\prime m ),\]
	for some constant $c_2$. It implies $ \sum_{i \in [m]} z_i^2 \leq 2 m \sigma^2$ with probability at least $1-\exp ( - c_2 \epsilon^\prime m )$.
	Combining the above inequality with $|\Scal_1| \ge (1 - \epsilon^\prime) m$, we obtain
	\begin{equation*}
		\sigma^2_{\Scal_1,\mu} = \frac{1}{\abs{\Scal_1}} \sum_{i \in [m]} z_i^2 \leq \frac{2 \sigma^2}{1-\epsilon^\prime},
	\end{equation*}
	with probability at least $1 - \exp ( - c_1 \epsilon^\prime m ) -\exp ( - c_2 \epsilon^\prime m )$.

	Finally, we establish the bound for $\mu_{\Scal_1}$, which is given by
	\begin{align*}
		\abs{\mu_{\Scal_1}} & \leq \abs{\mu_{\Scal_1} - \mu_{\Scal_2}} + \abs{\mu_{\Scal_2}} \\
		 & = \biggl|\sum_i \Bigl(\frac{1}{\abs{\sS_1}} - \frac{1}{m}\Bigr) z_i \biggr| +
		\biggl|\frac{1}{m} \sum_{i \in [m]} z_i \biggr| \\
		 & \leq \frac{2\sigma}{\sqrt{\epsilon^\prime}} \sum_i \Bigl|\frac{1}{\abs{\sS_1}} - \frac{1}{m}\Bigr| +
		\biggl|\frac{1}{m} \sum_{i \in [m]} z_i \biggr| \\
		 & \leq 4\sigma\sqrt{\epsilon^\prime} +
		\biggl|\frac{1}{m} \sum_{i \in [m]} z_i \biggr|,
	\end{align*}
	where the second inequality comes from $\abs{z_i} \leq 2\sigma/\sqrt{\epsilon^\prime}$ and the third inequality is due to $\abs{\Scal_1} \geq (1 - \epsilon^\prime) m$.
	Given that $\abs{\E[z_i]} = \abs{\E[x_i \indicator{\{\abs{x_i} > 2\sigma/\sqrt{\epsilon^\prime}\}}]} \leq \sigma \sqrt{\epsilon^\prime}$, $\Var(z_i) \leq \sigma^2$ and $\abs{z_i} \leq 2\sigma/\sqrt{\epsilon^\prime}$, by Bennett's inequality,
	\[
		\Pr\Bigl( \Big|\sum_{i \in [m]} (z_i - \E[z_i])\Big| \geq m \sigma \sqrt{\epsilon^\prime} \Bigr) \leq \exp ( - c_3 m\epsilon^\prime ),
	\]
	for some constant $c_3$, which indicates
	\[
		\Bigl|\frac{1}{m} \sum_{i \in [m]} z_i \Bigr| \leq 2 \sigma \sqrt{\epsilon^\prime},
	\]
	with probability at least $1-\exp(-c_3 \epsilon^\prime m)$.
	Therefore, we have $\abs{\mu_{\Scal_1}} \leq 6 \sigma\sqrt{\epsilon^\prime}$.

	Combining the above results, we obtain $|\mu_{\Scal_1}| \leq 6 \sigma \sqrt{\epsilon^\prime} $ and $\sigma^2_{\Scal_1,\mu} \le 2 \sigma^2/(1-\epsilon^\prime)$
	with probability at least $1 - \exp(-c_1 \epsilon^\prime m) + \exp(-c_2 \epsilon^\prime m) + \exp(-c_3 \epsilon^\prime m)$.
	Recalling that $\epsilon^\prime \ge C\log(1/\tau) / m$, we can choose $C$ sufficiently large so that $\tau > \exp(-c_1 \epsilon^\prime m) + \exp(-c_2 \epsilon^\prime m) + \exp(-c_3 \epsilon^\prime m)$.
\end{proof}

In Lemma~\ref{lemma:main_lemma_in_RB}, we prove that a stable subset still exists in the corrupted set with high probability.
\begin{lemma}\label{lemma:main_lemma_in_RB}
	Let $\Tcal$ be an $\epsilon$-corrupted version of the set $\Scal$ considered in Lemma~\ref{lemma:control_the_variance} such that $|\Tcal \cap \Scal| \geq (1-\epsilon) m$ and $\abs{\Tcal} = m$.
	By setting $\epsilon_1 = (1 + c) \epsilon + C \log(1/\tau) / m$ for any constant $c>0$ and $C > 0$ specified in Lemma~\ref{lemma:control_the_variance}, with probability at least $1-\tau$, there exists a subset $\Tcal^\prime \subset \Tcal$ with $\abs{\Tcal^\prime} \geq (1-\epsilon_1)m$ such that
	\begin{equation*}
		\abs{\mu_{\Tcal^\prime} - \mu} = O\biggl(\sigma \sqrt{\epsilon} + \sigma \sqrt{\frac{\log(1/\tau)}{m}}\biggr)\, \mbox{ and }\,
		\sigma^2_{\Tcal^\prime,\mu} = O(\sigma^2).
	\end{equation*}
	In other words, let the entries of the weight vector $\w \in \triangle_{m, \epsilon_1}$ be such that $w_i = 1/\abs{\Tcal^\prime}$ if $x_i \in \Tcal^\prime$ otherwise $0$, we have
	\begin{equation*}
		\abs{\mu_{\w,\Tcal} - \mu} = O\biggl(\sigma \sqrt{\epsilon} + \sigma \sqrt{\frac{\log(1/\tau)}{m}}\biggr)\, \mbox{ and }\,
		\sigma^2_{\w,\Tcal,\mu} = O(\sigma^2).
	\end{equation*}
\end{lemma}
\begin{proof}
	We can choose $\epsilon^\prime = c \epsilon + C\log(1/\tau)/m$ in Lemma~\ref{lemma:control_the_variance}.
	With probability at least $ 1 - \tau$, there exists a subset $\Scal_1 \subset \Scal$ with $ |\Scal_1| \geq (1-\epsilon^\prime)m $ such that $ |\mu_{\Scal_1} - \mu| = O(\delta_1) $ and $ \sigma_{\Scal_1,\mu}^2 =O(\delta_1^2/\epsilon^\prime)$ with $\delta_1 = \sigma \sqrt{\epsilon^\prime}$.
	Applying Lemma~\ref{lemma:entire_set_stability_to_subset_stability}, we have $ \Scal_1 $ is $ (\epsilon_2,\delta_2) $-stable with respect to $ \mu $, where $\epsilon_2 = \{\epsilon + \log(1/\tau) / m\} / (1 - \epsilon^\prime)$
	and $\delta_2 = O(\delta_1 ( 1 + \sqrt{\epsilon_2/\epsilon^\prime})) = O(\sigma \sqrt{\epsilon^\prime})$.

	We choose the desired set as $\Tcal^\prime = \Scal_1 \cap \Tcal$ and $\epsilon_1 = \epsilon + \epsilon^\prime = (1 + c) \epsilon + C \log(1/\tau) / m$.
	It satisfies that $\abs{\Tcal^\prime} \ge |\Scal_1| - \epsilon m \ge (1-\epsilon_2) \abs{\Scal_1}$ and $\abs{\Tcal^\prime} \ge (1 - \epsilon_1) m$.
	By the stability of $\Scal_1$, we have
	\begin{equation*}
		\abs{\mu_{\Tcal^\prime} - \mu} = O(\sigma \sqrt{\epsilon^\prime}) = O\biggl(\sigma \sqrt{\epsilon} + \sigma \sqrt{\frac{\log(1/\tau)}{m}}\biggr) \, \mbox{ and }\,
		\sigma^2_{\Tcal^\prime,\mu} = O(\epsilon_2^{-1}\delta_2^2) = O(\sigma^2).
	\end{equation*}
\end{proof}

Equipped with the lemmas above, we can now begin proving the core results of this section.
Note that $\Lcal_{j,\targetx}(\vq) = \Lcal(\targetx,\sum_{k = 1}^K q_k \sourcexcyhat)$ and $\Ecal_{j,\targetx}(\vq,\vq^\prime) = \Lcal_{j,\targetx} (\vq) - \Lcal_{j,\targetx} (\vq^\prime)$.
The difference between $\Lcal_{j,\targetx}(\vq)$ and $\Lcal_{j}(\vq)$ lies in whether $\targetx$ or $\targetxhat$ is used to compute the divergence.
The introduction of $\Lcal_{j,\targetx}(\vq)$ is motivated by the fact that they are independent of each other and closely approximate $\Lcal_{j}(\vq)$.
We assume that $\Lcal_{j,\targetx}(\vq)$ and $\Ecal_{j,\targetx}(\vq,\vq^\prime)$ satisfy the following assumptions.

\begin{assumption}\label{ass:lipschitz}
	For any $ j \in [m] $, $\Lcal_{j,\targetx} (\vq)$ and $\Lcal_\target(\vq)$ are Lipschitz continuous, with a common Lipschitz constant $L$, i.e., $ |\Lcal_{j,\targetx} (\vq_1) - \Lcal_{j,\targetx} (\vq_2)| \leq L \norm{\vq_1 - \vq_2}_{2} $ and $ |\Lcal_\target (\vq_1) - \Lcal_\target (\vq_2)| \leq L \norm{\vq_1 - \vq_2}_{2} $ for any $\vq_1$ and $\vq_2 \in \triangle_{K}$.
\end{assumption}

\begin{assumption}\label{ass:target_adaptive_convergence}
	There exists a sequence $a_{m, N} = o(1)$ such that the following inequalities hold with probability at least $1-\tau/2$.

	(i) The empirical excess risk satisfies that
	\[\max_{j \in [m]} \sup_{\vq_1,\vq_2 \in \triangle_{K}} \frac{\abs{\Ecal_{j}(\vq_1,\vq_2) - \Ecal_{j,\targetx}(\vq_1,\vq_2)}}{\norm{\vq_1 - \vq^\ast}_{2} + \norm{\vq_2 - \vq^\ast}_{2}} \le a_{m, N}.\]

	(ii) The empirical divergence satisfies that
	\[\max_{j \in [m]} \sup_{ \vq \in \triangle_{K}} \abs{\Lcal_{j}(\vq) - \Lcal_{j,\targetx}(\vq)} \leq a_{m, N}.\]
\end{assumption}
Assumption~\ref{ass:target_adaptive_convergence} measures the difference between $\targetx$ and $\targetxhat$. Therefore, the error bound $a_{m, N}$ mainly depends on $N$ and $m$.

\begin{assumption}\label{ass:adaptive_moment}
	There exists a set $\Ical \subset [m]$ with $\abs{\Ical} \geq (1-\epsilon) m$ such that $ \E [\Lcal_{j,\targetx} (\vq)] = \Lcal_\target(\vq) $ for any $j \in \Ical$ and $ \vq \in \triangle_{K} $.
	Furthermore, there exists a sequence $\sigma_n = o(1)$ such that the following inequalities hold.

	(i)
	\[\max_{j \in \Ical} \sup_{\vq_1,\vq_2 \in \triangle_{K}} \frac{\Var[\Ecal_{j,\targetx}(\vq_1,\vq_2)]}{(\norm{\vq_1 - \vq^\ast}_{2} + \norm{\vq_2 - \vq^\ast}_{2})^2}\leq \sigma^2_{n}.\]

	(ii)
	\[
		\max_{j \in \Ical} \sup_{\vq \in \triangle_{K}} \Var[\Lcal_{j,\targetx} (\vq)] \leq \sigma^2_{n}.
	\]
\end{assumption}

In the next section, we will verify the above assumptions for the MMD loss and specify the rate of the scaling sequence $\sigma_n$.
Lemma~\ref{lemma:sup_of_adaptive_RB} below provides bounds in Assumptions~\ref{ass:uniform_converge_excess_risk} and \ref{ass:uniform_converge_divergence} for the MWV estimator.

\begin{lemma} \label{lemma:sup_of_adaptive_RB}
	Choosing $\epsilon_h = (1 + c) \epsilon + C \log(1/\tau) / m$ for any constant $c>0$, and $C > 0$ specified in Lemma~\ref{lemma:control_the_variance}, the following results hold.

	(i) Under Assumption~\ref{ass:lipschitz}, \ref{ass:target_adaptive_convergence}(i) and \ref{ass:adaptive_moment}(i), with probability at least $1-\tau$, it follows
	\begin{align*}
		\sup_{\vq_1,\vq_2 \in \triangle_{K}} & \frac{\abs{\mathrm{MWV}_{\epsilon_h}(\{\Ecal_j(\vq_1,\vq_2)\}_{j \in [m]}) - \Ecal_{\target}(\vq_1,\vq_2)}}{\max(\norm{\vq_1 - \vq^\ast}_{2} + \norm{\vq_2 - \vq^\ast}_{2},r_n)} \lesssim r_n.
	\end{align*}

	(ii) Under Assumptions~\ref{ass:lipschitz},\ref{ass:target_adaptive_convergence}(ii) and \ref{ass:adaptive_moment}(ii), with probability at least $1-\tau$, it follows
	\[\sup_{\vq \in \triangle_{K}} \abs{\mathrm{MWV}_{\epsilon_h}(\{\Lcal_{j}(\vq)\}_{j \in [m]}) - \Lcal_\target(\vq)} \lesssim r_n,\]
	where $r_n = \sigma_n \sqrt{\epsilon} + \sigma_n \sqrt{\frac{\log(1/\tau) + K \log (Lm/\sigma_n) }{m}} + a_{m, N}$.
\end{lemma}

\begin{proof}
	We first prove Part (i).
	We can construct a $\gamma$-net of $\triangle^{K} \times \triangle^{K}$ with cardinality $ O(\gamma^{-2K}) $, where $\gamma$ will be determined later.
	For any $(\vq_1,\vq_2) \in \triangle^{K} \times \triangle^{K}$, there exist $(\vq_1^\prime , \vq_2^\prime)$ in this $\gamma$-net such that $\max(\norm{\vq_1 - \vq_1^\prime}_{2},\norm{\vq_2 - \vq_2^\prime}_{2}) \leq \gamma$.
	Let $\Zcal_1 = \{z_{j,\targetx}^\prime\}_{j \in [m]}$, $\Zcal_2 = \{z_{j,\targetx}\}_{j \in [m]}$ and $\Zcal_3 = \{z_{j}\}_{j \in [m]}$, where $z_{j,\targetx}^\prime = \Ecal_{j,\targetx}(\vq_1^\prime,\vq_2^\prime) - \Ecal_Q(\vq_1^\prime,\vq_2^\prime)$, $ z_{j,\targetx} = \Ecal_{j,\targetx}(\vq_1,\vq_2) - \Ecal_Q(\vq_1,\vq_2)$ and $ z_{j} = \Ecal_{j}(\vq_1,\vq_2) - \Ecal_Q(\vq_1,\vq_2)$, respectively.

	Under Assumption~\ref{ass:adaptive_moment}(i), we have $\E[z_{j,\targetx}^\prime] = 0$ and $\Var(z_{j,\targetx}^\prime) \leq \sigma_n^2 (\norm{\vq_1^\prime - \vq^\ast}_{2} + \norm{\vq_2^\prime - \vq^\ast}_{2})^2$.
	By Lemma~\ref{lemma:main_lemma_in_RB}, we can find a subset in $\Zcal_1$ and a corresponding weight $\w \in \triangle_{m, \epsilon_h}$ such that
	\begin{align*}
		|\mu_{\w,\Zcal_1}| & = O\Bigl(\delta_n (\norm{\vq_1^\prime - \vq^\ast}_{2} + \norm{\vq_2^\prime - \vq^\ast}_{2})\Bigr) \\
		\sigma^2_{\w,\Zcal_1,0} & = O\Bigl(\delta_n^2 (\norm{\vq_1^\prime - \vq^\ast}_{2} + \norm{\vq_2^\prime - \vq^\ast}_{2})^2/\epsilon_h\Bigr)
	\end{align*}
	with probability at least $1-\tilde{\tau}/2$ for a fixed $(\vq_1^\prime, \vq_2^\prime)$, where $\delta_n = O(\sigma_n \sqrt{\epsilon} + \sigma_n \sqrt{\log(1/\tilde{\tau})/m})$.

	Given the Lipschitz continuity of $\Lcal_{j,\targetx}(\vq)$, we have $ |z_{j,\targetx} - z_{j,\targetx}^\prime| \leq 2 L (\norm{\vq_1 - \vq_1^\prime}_{2} + \norm{\vq_2 - \vq_2^\prime}_{2})\leq 4L \gamma $ and $|\mu_{\w,\Zcal_2} - \mu_{\w,\Zcal_1} | \leq 4L \gamma$.
	Under Assumption~\ref{ass:target_adaptive_convergence}(i), we also have $ |z_{j,\targetx} - z_{j}| \leq a_{m,N} (\norm{\vq_1 - \vq^\ast}_{2} + \norm{\vq_2 - \vq^\ast}_{2})$ and $|\mu_{\w,\Zcal_2} - \mu_{\w,\Zcal_3} | \leq a_{m,N} (\norm{\vq_1 - \vq^\ast}_{2} + \norm{\vq_2 - \vq^\ast}_{2} )$ with probability at least $1-\tau/2$ uniformly for all $(\vq_1, \vq_2)$.
	By setting $\gamma = L^{-1}\delta_n r_n$,
	we obtain
	\begin{align*}
		|\mu_{\w,\Zcal_3}| & \leq |\mu_{\w,\Zcal_3} - \mu_{\w,\Zcal_2} |+|\mu_{\w,\Zcal_2} - \mu_{\w,\Zcal_1} |+ |\mu_{\w,\Zcal_1}| \\
		 & = O\Bigl(a_{m,N} (\norm{\vq_1 - \vq^\ast}_{2} + \norm{\vq_2 - \vq^\ast}_{2}) + \gamma + \delta_n (\norm{\vq_1^\prime - \vq^\ast}_{2} + \norm{\vq_2^\prime - \vq^\ast}_{2}) \Bigr) \\
		 & = O\Bigl((\delta_n + a_{m, N}) \max(\norm{\vq_1 - \vq^\ast}_{2} + \norm{\vq_2 - \vq^\ast}_{2},r_n)\Bigr),
	\end{align*}
	and
	\begin{align*}
		\sigma^2_{\w,\Zcal_3,0} & = \sum_{j \in [m]} w_{j}z_{j}^2 \\
		 & \leq 3 \sum_{j \in [m]} w_{j}\bigl\{ (z_{j} - z_{j,\targetx})^2 + (z_{j,\targetx} - z_{j,\targetx}^\prime)^2 + (z_{j,\targetx}^\prime)^2 \bigr\} \\
		 & = O\Bigl(a_{m, N}^2 (\norm{\vq_1 - \vq^\ast}_{2} + \norm{\vq_2 - \vq^\ast}_{2})^2 + \gamma^2 + \delta_n^2 (\norm{\vq_1^\prime - \vq^\ast}_{2} + \norm{\vq_2^\prime - \vq^\ast}_{2})^2/\epsilon_h \Bigr) \\
		 & = O\Bigl((\delta_n + a_{m,N})^2 \max(\norm{\vq_1 - \vq^\ast}_{2} + \norm{\vq_2 - \vq^\ast}_{2},r_n)^2/\epsilon_h\Bigr).
	\end{align*}
	By Lemma~\ref{lemma:stability_on_corrupted_set}, we conclude that with probability at least $1 - (\tau + \tilde{\tau}) / 2$,
	\[\frac{\abs{\mathrm{MWV}_{\epsilon_h}(\{\Ecal_j(\vq_1,\vq_2)\}_{j \in [m]}) - \Ecal_{\target}(\vq_1,\vq_2)}}	{\max(\norm{\vq_1 - \vq^\ast}_{2} + \norm{\vq_2 - \vq^\ast}_{2},r_n)} \leq \sigma_n \sqrt{\epsilon} + \sigma_n \sqrt{\frac{\log(1/\tilde{\tau})}{m}} + a_{m, N},\]
	for any $\vq_1, \vq_2$ such that $\max(\norm{\vq_1 - \vq_1^\prime}_{2},\norm{\vq_2 - \vq_2^\prime}_{2}) \leq \gamma$.

	By choosing $\tilde{\tau} = \gamma^{2K} \tau$ and taking the union bound over all the $(\vq_1^\prime, \vq_2^\prime)$ in the $\gamma$-net, we obtain
	\begin{align*}
		 & \sup_{\vq_1, \vq_2 \in \triangle_{K}} \frac{\abs{\mathrm{MWV}_{\epsilon_h}(\{\Ecal_j(\vq_1,\vq_2)\}_{j \in [m]}) - \Ecal_{\target}(\vq_1,\vq_2)}}	{\max(\norm{\vq_1 - \vq^\ast}_{2} + \norm{\vq_2 - \vq^\ast}_{2},r_n)} \\
		\lesssim & \sigma_n \sqrt{\epsilon} + \sigma_n \sqrt{\frac{\log(1/\tau) + K \log (L m/\sigma_n) }{m}} + a_{m, N},
	\end{align*}
	which holds with probability at least $1-\tau$.

	We then study Part (ii).
	Similarly, we can construct a $ \gamma $-net of $ \triangle_{K} $ with cardinality $ O(1/\gamma^K) $, where $ \gamma = L^{-1}(\sigma_n \sqrt{\epsilon} + \sigma_n \sqrt{\log(1/\tau)/m} )$.
	For any $ \vq \in \triangle_{K} $, there exists $\vq^\prime$ in this $ \gamma $-net such that $ \norm{\vq^\prime-\vq}_{2} \leq \gamma $.

	Let $\Zcal_1 = \{z_{j,\targetx}^\prime\}_{j \in [m]}$, $\Zcal_2 = \{z_{j,\targetx}\}_{j \in [m]}$ and $\Zcal_3 = \{z_{j}\}$, where $z_{j,\targetx}^\prime = \Lcal_{j,\targetx}(\vq^\prime) - \Lcal_{\target} (\vq^\prime)$, $ z_{j,\targetx} = \Lcal_{j,\targetx}(\vq) - \Lcal_Q(\vq)$ and $ z_{j} = \Lcal_{j}(\vq) - \Lcal_Q(\vq)$, respectively.
	Under Assumption~\ref{ass:adaptive_moment}(ii), we have $\E[z_{j,\targetx}^\prime] = 0$ and $\Var(z_{j,\targetx}^\prime) \leq \sigma_n^2$.
	By Lemma~\ref{lemma:main_lemma_in_RB}, we can find a subset in $\Zcal_1$ and a corresponding weight $\w \in \triangle_{m, \epsilon_h}$ such that with probability at least $1-\tilde{\tau}/2$, for a fixed $\vq'$,
	\begin{align*}
		|\mu_{\w,\Zcal_1}| = O(\delta_n) , \,
		\sigma^2_{\w,\Zcal_1,0} = O(\delta_n^2/\epsilon_h),
	\end{align*}
	where $\delta_n = O(\sigma_n \sqrt{\epsilon} + \sigma_n \sqrt{\log(1/\tilde{\tau})/m})$.

	Given the Lipschitz continuity of $\Lcal_{j,\targetx}(\vq)$, we have $ |z_{j,\targetx} - z_{j,\targetx}^\prime| \leq 2L\norm{\vq - \vq^\prime}_{2}\leq 2L \gamma $ and $|\mu_{\w,\Zcal_2} - \mu_{\w,\Zcal_1} | \leq 2L \gamma$.
	Under Assumption~\ref{ass:target_adaptive_convergence}(ii), we also have $ |z_{j,\targetx} - z_{j}| \leq a_{m,N}$ and $|\mu_{\w,\Zcal_2} - \mu_{\w,\Zcal_3} | \leq a_{m,N}$ with probability at least $1-\tau/2$ uniformly for all $\vq$.
	Therefore, we have,
	\begin{align*}
		|\mu_{\w,\Zcal_3}| & \leq |\mu_{\w,\Zcal_3} - \mu_{\w,\Zcal_2} |+|\mu_{\w,\Zcal_2} - \mu_{\w,\Zcal_1} |+ |\mu_{\w,\Zcal_1}| \\
		 & = O(a_{m, N} + \gamma + \delta_n ) = O(\delta_n + a_{m, N}),
	\end{align*}
	and
	\begin{align*}
		\sigma^2_{\w,\Zcal_3,0} & = \sum_{j \in [m]} w_{j}z_{j}^2 \\
		 & \leq 3 \sum_{j \in [m]} w_{j}\bigl\{ (z_{j} - z_{j,\targetx})^2 + (z_{j,\targetx} - z_{j,\targetx}^\prime)^2 + (z_{j,\targetx}^\prime)^2 \bigr\} \\
		 & = O(a_{m, N}^2 + \gamma^2 + \delta_n^2/\epsilon_h ) \\
		 & = O\Bigl((\delta_n + a_{m,N})^2 /\epsilon_h\Bigr).
	\end{align*}
	By Lemma~\ref{lemma:stability_on_corrupted_set}, we conclude that with probability at least $1 - (\tau + \tilde{\tau}) / 2$,
	\[\abs{\mathrm{MWV}_{\epsilon_h}(\{\Lcal_j(\vq)\}_{j \in [m]}) - \Lcal_{\target}(\vq)} \lesssim \sigma_n \sqrt{\epsilon} + \sigma_n \sqrt{\frac{\log(1/\tilde{\tau})}{m}} + a_{m, N},\]
	for any $\vq$ such that $\lVert \vq - \vq^\prime \rVert_2 \le \gamma$.
	By choosing $\tilde{\tau} = \gamma^{K} \tau$ and taking the union bound over all $\vq^\prime$ in the $\gamma$-net, we obtain that with probability at least $1-\tau$,
	\[\sup_{\vq \in \triangle_{K}} \abs{\mathrm{MWV}_{\epsilon_h}(\{\Lcal_{j}(\vq)\}_{j \in [m]}) - \Lcal_\target(\vq)} \lesssim\sigma_n \sqrt{\epsilon} + \sigma_n \sqrt{\frac{\log(1/\tau) + K \log (Lm/\sigma_n) }{m}} + a_{m, N}.\]
\end{proof}

\subsection{Proof of Lemma~\ref{lemma:upper_var}}\label{sec:proof_upper_var}
\begin{proof}
Recall that $\Lcal_{j,\targetx}(\vq)$ is defined as
	\begin{equation*}
		\Lcal_{j,\targetx}(\vq) = \MMDhat\Bigl(\targetx,\sum_{k=1}^{K} q_k \hat{P}_{j,X|Y = k}\Bigr) = \vq^\top \hat{\mA}_{j} \vq - 2 \vq^\top \bar{\vb}_{j},
	\end{equation*}
	where $\hat{\mA}_{j}$ aligns with Eq.~(\ref{equation:MMD}) and
	\begin{align}\label{eq:b_bar}
		\bar{b}_{j,k} & = \frac{1}{n_{j,k}} \sum_{x \in \Dcal_{j,k}} \mu_\targetx(x), \ \mu_\targetx(\cdot) = \E_{x \sim \targetx}[\kernel(x,\cdot)].
	\end{align}

First, we derive an upper bound for $\Var[\Ecal_{j,\targetx}(\vq_1,\vq_2)]$ to establish the first result in Lemma~\ref{lemma:upper_var} under Assumptions~\ref{ass:sample_balance} and \ref{ass:kernel}. For simplicity, we suppress the index $j$ in this proof, which denotes the source index in the terms $\widehat{\mA}_j$, $\bar{\vb}_j$, $\Dcal_{j,k}$, $n_{j,k}$, and related expressions.
	For example, $\widehat{\mA} = (\Ahat_{k,k^\prime})_{k,k^\prime \in [m]}$, where $\Ahat_{k,k^\prime} = \Ahat_{j,k,k^\prime}$.

	Observe that
	\begin{align*}
		\Ecal_{j,\targetx}(\vq_1,\vq_2) & = (\vq_1 + \vq_2)^\top \widehat{\mA} (\vq_1 - \vq_2) - \bar{\vb}^\top (\vq_1 - \vq_2) \\
		 & = \sum_{i,h \in [K]} v_{i,+} v_{h,-} \Ahat_{i,h} - \sum_{i \in [K]} \bbar_i v_{i,-},
	\end{align*}
	where $\vv_+ = (v_{i,+})_{i\in[K]} = (q_{1,i} + q_{2,i})_{i \in [K]}$ and $\vv_- = (v_{i,-})_{i\in[K]} = (q_{1,i} - q_{2,i})_{i \in [K]}$.
	The variance can be bounded as
	\begin{equation}\label{eq:var_excess_decomp}
		\Var[\Ecal_{j,\targetx}(\vq_1,\vq_2)] = O\biggl(\Var\Bigl( \sum_{i,h \in [K]} v_{i,+} v_{h,-} \Ahat_{i,h}\Bigr)\biggr) + O\biggl(\Var\Bigl( \sum_{i \in [K]} \bbar_i v_{i,-}\Bigr)\biggr).
	\end{equation}

	By the definition of $\bar{\vb}$ in Eq.~(\ref{eq:b_bar}), the covariances between $\bbar_i$ and $\bbar_h$ for $i,h \in [K]$ satisfy the following properties:
	\begin{itemize}
		\item For $i \neq h$, $\Cov(\bbar_i,\bbar_h) = 0$.
		\item For $i = h$, $\Cov(\bbar_i,\bbar_i) = n_i^{-2} \sum_{x_i \in \Dcal_i} \Cov(\mu_\targetx(x_i), \mu_\targetx(x_i)) = O(n_i^{-1}) = O((n \psi)^{-1})$.
	\end{itemize}
	Therefore, for the second term in Eq.~(\ref{eq:var_excess_decomp}), we obtain
	\begin{equation*}
		\Var\Bigl(\sum_{i \in [K]} \bbar_i v_{i,-}\Bigr) = \sum_{i \in [K]} v_{i,-}^2 \Var(\bbar_i) = O\biggl(\frac{\norm{\vv_{-}}_{2}^2}{n\psi}\biggr).
	\end{equation*}

	It remains to analyze the first term in Eq.~(\ref{eq:var_excess_decomp}).
	Specifically, we aim to bound
	\begin{align*}
		 & \Var\Bigl(\sum_{i,h \in [K]} v_{i,+} v_{h,-} \Ahat_{i,h}\Bigr) \\
		= & \sum_{i,h,g,k \in [K]} v_{i,+} v_{h,-} v_{g,+} v_{k,-} \Cov(\Ahat_{i,h},\Ahat_{g,k}) \\
		= & \sum_{i,h,k \in [K]} v_{i,+}^2 v_{h,-} v_{k,-} \Cov(\Ahat_{i,h},\Ahat_{i,k}) + \sum_{\substack{i,h,g \in [K] \\ i \neq g}} v_{h,-}^2 v_{i, +} v_{g, +} \Cov(\Ahat_{i, h},\Ahat_{g, h})\\
		 & + 2\sum_{\substack{i,h,g \in [K] \\ i \neq h, i \neq g}} v_{i,+} v_{h,-} v_{g,+} v_{i,-} \Cov(\Ahat_{i,h},\Ahat_{g,i}) + 2\sum_{\substack{i,g,k \in [K] \\ i \neq g, i \neq k}} v_{i,+} v_{i,-} v_{g,+} v_{k,-} \Cov(\Ahat_{i,i},\Ahat_{g, k}) \\
		 & - \sum_{i \neq h \in [K]} v_{i, +}v_{h, -} v_{h, +} v_{i, -} \Cov(\Ahat_{i, h}, \Ahat_{h, i}) - \sum_{i \neq g \in [K]} v_{i, +}v_{i, -} v_{g, +} v_{g, -} \Cov(\Ahat_{i, i}, \Ahat_{g, g}) \\
		 & + \sum_{i, h, g, k \text{ distinct}} v_{i,+} v_{h,-} v_{g,+} v_{k,-} \Cov(\Ahat_{i,h},\Ahat_{g,k})
	\end{align*}
	By noting that $\Cov(\Ahat_{i,h},\Ahat_{g,k}) = 0$ if $\{i,h\} \cap \{g,k\} = \emptyset$, the expression simplifies to
	\begin{align*}
		 & \Var\Bigl(\sum_{i,h \in [K]} v_{i,+} v_{h,-} \Ahat_{i,h}\Bigr) \\
		= & \underbrace{\sum_{i,h,k \in [K]} v_{i,+}^2 v_{h,-} v_{k,-} \Cov(\Ahat_{i,h},\Ahat_{i,k})}_{\Rmnum{1}} + \underbrace{\sum_{\substack{i,h,g \in [K] \\ i \neq g}} v_{h,-}^2 v_{i, +} v_{g, +} \Cov(\Ahat_{i, h},\Ahat_{g, h})}_{\Rmnum{2}}\\
		 & + \underbrace{2\sum_{\substack{i,h,g \in [K] \\ i \neq h, i \neq g}} v_{i,+} v_{h,-} v_{g,+} v_{i,-} \Cov(\Ahat_{i,h},\Ahat_{g,i})}_{\Rmnum{3}} - \underbrace{\sum_{i \neq h \in [K]} v_{i, +}v_{h, -} v_{h, +} v_{i, -} \Cov(\Ahat_{i, h}, \Ahat_{h, i})}_{\Rmnum{4}}
	\end{align*}

	We present the following proposition concerning the covariance $\Cov(\Ahat_{i, h}, \Ahat_{g, k})$, with the proof provided at the end of this subsection.
	\begin{proposition}\label{prop:bound_cov}
		Under Assumption~\ref{ass:sample_balance} and Assumption~\ref{ass:kernel}(i), we have
		\begin{equation*}
			\sup_{i,h,g,k \in [K]} \abs{\Cov(\Ahat_{i, h}, \Ahat_{g, k})} = \frac{5 \kappa^2}{n \psi}.
		\end{equation*}
	\end{proposition}

	Utilizing Proposition~\ref{prop:bound_cov}, we derive the bounds for the aforementioned four terms.
	we define $\mC_i = (C_{i,h,k})_{h,k \in [K]} = (\Cov(\Ahat_{i,h},\Ahat_{i,k}))_{h,k \in [K]}$ and $\mD_i = (D_{i,h,k})_{h,k \in [k]} = \Cov(\phi_i)$, where $\phi_i = (\Kcal(x_{i1}, x_{k2}))_{k \in [K]}$ and $\Cov(\phi_i)$ are defined in Assumption~\ref{ass:kernel}.
	For the first summation, denoted as $\Rmnum{1}$, we have
	\begin{equation*}
		\Rmnum{1} = \sum_{i \in [K]} v_{i,+}^2 \sum_{h,k} v_{h,-} v_{k,-} \Cov(\Ahat_{i,h},\Ahat_{i,k}) = \sum_{i \in [K]} v_{i,+}^2 \vv_{-}^\top \mC_i \vv_{-}.
	\end{equation*}
	For distinct $i,h,k$,
	\begin{equation*}
		C_{i,h,k} = \Cov(\Ahat_{i,h},\Ahat_{i,k}) = \frac{1}{n_i} \Cov(\kernel(x_{i1},x_{h2}),\kernel(x_{i1},x_{k2})) = \frac{1}{n_i} D_{i,h,k}.
	\end{equation*}
	For $i = h \neq k$,
	\begin{equation*}
		C_{i,i,k} = C_{i,k,i} = \Cov(\Ahat_{i,i},\Ahat_{i,k}) = \frac{2}{n_i} \Cov(\kernel(x_{i1},x_{i2}),\kernel(x_{i1},x_{k2})) = \frac{2}{n_i}D_{i,i,k}.
	\end{equation*}
	For $i = h = k$, $C_{i,i,i} = \Cov(\Ahat_{i,i},\Ahat_{i,i}) = O((n\psi)^{-1})$.
	Under Assumption~\ref{ass:kernel}, $\lVert \mD_i \rVert_{\mathrm{op}} = \lVert \Cov(\phi_i) \rVert_{\mathrm{op}} \le C$.
	Therefore,
	\begin{equation*}
		\sup_{i \in [K]}\lVert \mC_i \rVert_{\mathrm{op}} \le \sup_{i \in [K]} \frac{1}{n_i}\lVert \mD_i \rVert_{\mathrm{op}} + O\Bigl(\frac{1}{n\psi}\Bigr) = O\Bigl(\frac{1}{n\psi}\Bigr).
	\end{equation*}
	It implies that
	\begin{equation*}
		\Rmnum{1} = \sum_{i \in [K]} v_{i,+}^2 \vv_{-}^\top \mC_i \vv_{-} \lesssim \frac{1}{n \psi} \norm{\vv_{+}}_2^2 \norm{\vv_{-}}_2^2 \le \frac{4}{n \psi} \norm{\vv_{-}}_2^2.
	\end{equation*}

	For the second summation $\Rmnum{2}$,
	\begin{equation*}
		\Rmnum{2} \lesssim \frac{1}{n \psi} \norm{\vv_{-}}_2^2 \Bigl(\sum_{i \in [K]} v_{i,+}\Bigr)^2 = \frac{4}{n \psi} \norm{\vv_{-}}_2^2.
	\end{equation*}

	For the third summation $\Rmnum{3}$,
	denote $\vv_{+,i\rightarrow 0}$ and $\vv_{-, i\rightarrow 0}$ as vectors that set the $i$-th entries of $\vv_{+}$ and $\vv_{-}$ to zero, respectively.
	We have
	\begin{align*}
		\Rmnum{3} & = 2 \sum_i v_{i,+} v_{i,-} \sum_{h \neq i, g \neq i} v_{h,-} v_{g,+} \Cov(\Ahat_{i,h},\Ahat_{g,i}) \\
		 & = 2 \sum_i v_{i,+} v_{i,-} \frac{\vv_{-,i\rightarrow 0} ^\top \mD_{i} \vv_{+,i\rightarrow 0}}{n_i} \\
		 & \leq 2 \sum_i |v_{i,+} v_{i,-}| \frac{\lVert \mD_{i} \rVert_{\mathrm{op}} \norm{\vv_{-,i\rightarrow 0}}_{2} \norm{\vv_{+,i\rightarrow 0}}_{2}}{n_i} \\
		 & \leq \frac{2}{n \psi} \sum_i |v_{i,+} v_{i,-}| \lVert \mD_{i} \rVert_{\mathrm{op}} \norm{\vv_{-}}_{2} \norm{\vv_{+}}_{2} \\
		 & \lesssim \frac{1}{n \psi} \norm{\vv_{-}}_{2}^2 \norm{\vv_{+}}_{2}^2 \leq \frac{4 \norm{\vv_{-}}_{2}^2}{n \psi},
	\end{align*}
	in which the third inequality arises from the application of the Cauchy-Schwarz inequality, alongside the fact that $\sup_{i \in [K]} \lVert \mD_{i} \rVert_{\mathrm{op}} \leq C$, as stipulated in Assumption~\ref{ass:kernel}.

	For the last summation $\Rmnum{4}$,
	\begin{equation*}
		\Rmnum{4} \le \frac{1}{n \psi} \sum_{i \neq h \in [K]} |v_{i, +}v_{i, -}| |v_{h, +} v_{h, -}| \le \frac{1}{n \psi} \Bigl(\sum_{i \in [K]} |v_{i, +}v_{i, -}|\Bigr)^2 \le \frac{1}{n \psi} \norm{\vv_{+}}_2^2 \norm{\vv_{-}}_2^2 \le \frac{4}{n \psi} \norm{\vv_{-}}_2^2.
	\end{equation*}

	Combining $\Rmnum{1}$--$\Rmnum{4}$, we conclude that
	\begin{equation*}
		\Var\Bigl(\sum_{i,h \in [K]} v_{i,+} v_{h,-} \Ahat_{i,h}\Bigr) = O\biggl(\frac{\norm{\vv_{-}}_{2}^2}{n\psi}\biggr),
	\end{equation*}
	and
	\begin{equation*}
		\Var[\Ecal_{j,\targetx}(\vq_1,\vq_2)] = O\biggl(\frac{\norm{\vv_{-}}_{2}^2}{n\psi}\biggr) = O\biggl(\frac{\norm{\vq_1 - \vq_2}_2^2}{n \psi}\biggr) = O\biggl(\frac{(\norm{\vq_1-\vq^\ast}_{2} + \norm{\vq_2-\vq^\ast}_{2})^2}{n\psi}\biggr),
	\end{equation*}
	uniformly for $j \in [m]$ and $\vq_1,\vq_2 \in \triangle_K$.

        Next, we derive an upper bound for $\Var[\Lcal_{j,\targetx}(\vq)]$ to establish the second result in Lemma~\ref{lemma:upper_var}. Note that
	\begin{equation*}
		\Lcal_{j,\targetx}(\vq) = \vq^\top \widehat{\mA} \vq - \bar{\vb}^\top \vq = \sum_{i,h \in [K]} q_i q_h \Ahat_{i,h} - \sum_{i \in [K]} \bbar_i q_i.
	\end{equation*}
	Its variance is bounded by
	\begin{equation*}
		O\biggl(\Var\Bigl( \sum_{i,h \in [K]} q_i q_h \Ahat_{i,h} \Bigr)\biggr) + O\biggl(\Var\Bigl(\sum_{i \in [K]} \bbar_i q_i\Bigr)\biggr).
	\end{equation*}
	The first part satisfies that
	\begin{align*}
		 & \Var\Bigl(\sum_{i,h \in [K]} q_i q_h \Ahat_{i,h}\Bigr) \\
		= & \sum_{i,h,g,k \in [K]} q_i q_h q_g q_k \Cov(\Ahat_{i,h},\Ahat_{g,k}) \\
		= & O\Bigl(\frac{1}{n\psi} \sum_{i,h,g,k \in [K]} q_i q_h q_g q_k \Bigr) = O\Bigl(\frac{1}{n\psi}\Bigr),
	\end{align*}
	where the the second equality comes from Proposition~\ref{prop:bound_cov}, and the last equality uses the fact that $\sum_i q_i = 1$.

	The second part can be written as
	\begin{equation*}
		\Var\Bigl(\sum_{i \in [K]} \bbar_i q_i\Bigr) = \sum_{i \in [K]} q_i^2 \Var(\bbar_i) = O\Bigl(\frac{1}{n\psi}\Bigr),
	\end{equation*}
	where we use the fact that $\Cov(\bbar_i,\bbar_h) = 0$ for $i \neq h$ and $\Var(\bbar_i) = O(n_i^{-1}) = O((n \psi)^{-1})$.

	Note that the above two bounds hold uniformly for all $j \in [m]$ and $\vq \in \Delta_K$.
	Combining the above results, we have
	\begin{equation*}
		\sup_{j \in [m], \vq \in \Delta_K} \Var[\Lcal_{j,\targetx}(\vq)] = O\Bigl(\frac{1}{n\psi}\Bigr).
	\end{equation*}
\end{proof}

\begin{proof}[Proof of Proposition~\ref{prop:bound_cov}]
	Here, we denote $x_{ki}$ as the $i$-th sample in $\Dcal_k$.
	The covariance between $\Ahat_{i,h}$ and $\Ahat_{g,k}$ for $i,h,g,k \in [K]$ is given as follows.
	\begin{itemize}
		\item When $\{i,h\} \cap \{g,k\} = \emptyset$, $\Cov(\Ahat_{i,h},\Ahat_{g,k}) = 0$.
		\item When $i$, $h$ and $k$ are distinct and $g = i$,
		 \begin{align*}
			 & \Cov(\Ahat_{i,h},\Ahat_{i,k}) \\
			 = & \frac{1}{n_i^2 n_h n_k} \sum_{x_i \in \Dcal_i} \sum_{x_h \in \Dcal_h} \sum_{x_i^\prime \in \Dcal_i} \sum_{x_k \in \Dcal_k} \Cov(\kernel(x_i,x_h),\kernel(x_i^\prime,x_k)) \\
			 = & \frac{1}{n_i} \Cov(\kernel(x_{i1},x_{h1}),\kernel(x_{i1},x_{k1})).
		 \end{align*}
		 Therefore, $\abs{ \Cov(\Ahat_{i,h},\Ahat_{i,k}) } \le \kappa^2 / (n \psi)$.
		\item When $i = g$, $h = k$ and $i \neq h$,
		 \begin{align*}
			 & \Cov(\Ahat_{i,h},\Ahat_{i,h}) \\
			 = & \frac{1}{n_i^2 n_h^2} \sum_{x_i \in \Dcal_i} \sum_{x_h \in \Dcal_h} \sum_{x_i^\prime \in \Dcal_i} \sum_{x_h^\prime \in \Dcal_h} \Cov(\kernel(x_i,x_h),\kernel(x_i^\prime,x_h^\prime)) \\
			 = & \frac{1}{n_i^2 n_h^2} \sum_{x_i \in \Dcal_i} \sum_{x_h \in \Dcal_h} \Var(\kernel(x_i,x_h)) \\
			 & + \frac{1}{n_i^2 n_h^2} \sum_{x_i \in \Dcal_i} \sum_{x_h \in \Dcal_h} \sum_{\substack{x_h^\prime \in \Dcal_h, \\ x_h^\prime \neq x_h}} \Cov(\kernel(x_i,x_h),\kernel(x_i,x_h^\prime))\\
			 & +\frac{1}{n_i^2 n_h^2} \sum_{x_i \in \Dcal_i} \sum_{x_h \in \Dcal_h} \sum_{\substack{x_i^\prime \in \Dcal_i, \\ x_i^\prime \neq x_i}} \Cov(\kernel(x_i,x_h),\kernel(x_i^\prime,x_h))\\
			 = & \frac{1}{n_i n_h} \Var(\kernel(x_{i1},x_{h1})) + \frac{1}{n_i} \Cov(\kernel(x_{i1},x_{h1}),\kernel(x_{i1},x_{h2})) \\
			 & + \frac{1}{n_h} \Cov(\kernel(x_{i1},x_{h1}),\kernel(x_{i2},x_{h1})).
		 \end{align*}
		 Therefore, $\abs{ \Cov(\Ahat_{i,h},\Ahat_{i,h}) } \le 3 \kappa^2 / (n \psi)$.
		\item When $i = g = h$ and $i \neq k$,
		 \begin{align*}
			 & \Cov(\Ahat_{i,i},\Ahat_{i,k}) \\
			 = & \frac{1}{n_i^2 (n_i - 1) n_k} \sum_{x_i \in \Dcal_i} \sum_{\substack{x_i^\prime \in \Dcal_i \\ x_i^\prime \neq x_i}} \sum_{x_i^{\prime \prime} \in \Dcal_i} \sum_{x_k \in \Dcal_k} \Cov(\kernel(x_i,x_i^\prime),\kernel(x_i^{\prime\prime},x_k))\\
			 = & \frac{2}{n_i} \Cov(\kernel(x_{i1},x_{i2}),\kernel(x_{i1},x_{k1})) = O\Bigl(\frac{1}{n\psi}\Bigr).
		 \end{align*}
		 Therefore, $\abs{ \Cov(\Ahat_{i,i},\Ahat_{i,k}) } \le 2 \kappa^2 / (n \psi)$.

		\item When $i = g = h = k$,
		 \begin{align*}
			 & \Cov(\Ahat_{i,i},\Ahat_{i,i}) \\
			 = & \frac{1}{n_i^2 (n_i -1)^2} \sum_{x_i \in \Dcal_i} \sum_{\substack{x_i^\prime \in \Dcal_i \\ x_i^\prime \neq x_i}} \sum_{x_i^{\prime\prime} \in \Dcal_i} \sum_{\substack{x_i^{\prime\prime\prime} \in \Dcal_i\\ x_i^{\prime\prime\prime} \neq x_i^{\prime\prime}}} \Cov(\kernel(x_i,x_i^\prime),\kernel(x_i^{\prime\prime},x_i^{\prime\prime\prime}))\\
			 = & \frac{2}{n_i^2 (n_i -1)^2} \sum_{x_i \in \Dcal_i} \sum_{\substack{x_i^\prime \in \Dcal_i \\ x_i^\prime \neq x_i}} \Var(\kernel(x_i,x_i^\prime)) \\
			 & +\frac{4}{n_i^2 (n_i -1)^2} \sum_{x_i \in \Dcal_i} \sum_{\substack{x_i^\prime \in \Dcal_i \\ x_i^\prime \neq x_i}} \sum_{\substack{x_i^{\prime\prime} \in \Dcal_i \\ x_i^{\prime\prime} \neq x_i \\ x_i^{\prime\prime} \neq x_i^\prime}} \Cov(\kernel(x_i,x_i^\prime),\kernel(x_i,x_i^{\prime\prime}))\\
			 = & \frac{2}{n_i (n_i -1)} \Var(\kernel(x_{i1},x_{i2})) +
			 \frac{4(n_i-2)}{n_i(n_i-1)} \Cov(\kernel(x_{i1},x_{i2}),\kernel(x_{i1},x_{i3})) \\
			 = & O\Bigl(\frac{1}{n\psi}\Bigr).
		 \end{align*}
		 Therefore, $\abs{ \Cov(\Ahat_{i,i},\Ahat_{i,i}) } \le 5 \kappa^2 / (n \psi)$.
	\end{itemize}
\end{proof}

\subsection{Proof of Lemma~\ref{lemma:uniform_bound_mwv}}
\begin{proof}
    To prove Lemma~\ref{lemma:uniform_bound_mwv}, we mainly leverage Lemma~\ref{lemma:sup_of_adaptive_RB} and specify the parameters $L$, $a_{m,N}$ and $\sigma_n$ in Assumptions~\ref{ass:lipschitz}, \ref{ass:target_adaptive_convergence}, and \ref{ass:adaptive_moment}, respectively.

    {\noindent\bf Verify Assumption~\ref{ass:lipschitz}.} With the boundedness condition of $\kernel$, we have $\lVert \hat{\mA}_{j} \rVert_{\mathrm{op}}\leq \kappa K$, $\lVert \mA \rVert_{\mathrm{op}} \leq \kappa K$, $\norm{\bar{\vb}_{j}}_{2} \leq \kappa \sqrt{K}$ and $\norm{\vb}_{2} \leq \kappa \sqrt{K}$.
	It implies that $ \Lcal_{j,\targetx} (\vq) $ and $\Lcal_{\target}(\vq)$ are Lipschitz continuous, with a Lipschitz constant $L \le \kappa(4K + 2\sqrt{K})$.
	Therefore, Assumption~\ref{ass:lipschitz} holds with $L = O(K)$.

		{\noindent\bf Verify Assumption~\ref{ass:target_adaptive_convergence}.} For any $j \in [m]$, condition on $\Dcal_{j}$, by McDiarmid’s inequality,
	\begin{align*}
		\Pr\biggl(\Bigabs{\hat{b}_{j,k} - \bar{b}_{j,k}}>t \biggl| \Dcal_{j}\biggr) & = \Pr\biggl(\Bigabs{\frac{1}{N} \sum_{i \in [N]} \frac{1}{n_{j,k}} \sum_{x \in \Dcal_{j,k}}\kernel(x_i,x) - \bar{b}_{j,k}}>t\biggl| \Dcal_{j}\biggr) \\
		 & \leq 2\exp\Bigl(-\frac{Nt^2}{2\kappa^2}\Bigr).
	\end{align*}
	By marginalizing over the randomness of $\Dcal_{j}$, we have
	\[\Pr\Bigl(\bigabs{\hat{b}_{j,k} - \bar{b}_{j,k}}>t\Bigr) \leq 2\exp\Bigl(-\frac{Nt^2}{2\kappa^2}\Bigr).\]
	Therefore, by taking the union bound over $j \in [m]$ and $k \in [K]$, we have with probability at least $1-\tau/2$,
	\begin{equation}\label{equation:upper_bound_of_b}
		\sup_{j \in [m], k \in [K]}\bigabs{\hat{b}_{j,k} - \bar{b}_{j,k}} = O\biggl(\kappa \sqrt{\frac{\log(mK/\tau)}{N}}\biggr),
	\end{equation}
	which indicates that with probability at least $1 - \tau / 2$, uniformly for all $j \in [m]$ and all $\vq_1, \vq_2 \in \triangle_{K}$
	\begin{align*}
		 & \Ecal_{j}(\vq_1,\vq_2) - \Ecal_{j,\targetx}(\vq_1,\vq_2) \\
		= & \Lcal_{j}(\vq_1) - \Lcal_{j}(\vq_2) - \Lcal_{j,\targetx}(\vq_1) + \Lcal_{j,\targetx}(\vq_2) \\
		= & 2 \abs{ (\vq_1 - \vq_2)^\top (\bar{\vb}_{j} - \hat{\vb}_{j})} \\
		\leq & 2 \norm{\vq_1 - \vq_2}_{1} \sup_{j \in [m], k \in [K]}\bigabs{\hat{b}_{j,k} - \bar{b}_{j,k}} \\
		= & O \biggl(\sqrt{K} \norm{\vq_1 - \vq_2}_{2} \kappa \sqrt{\frac{\log(mK/\tau)}{N}} \biggr) \\
		= & O \biggl(\kappa\sqrt{\frac{K\log(mK/\tau)}{N}} (\norm{\vq_1 - \vq^\ast}_{2} + \norm{\vq_2 - \vq^\ast}_{2}) \biggr).
	\end{align*}
	Thus, Assumption~\ref{ass:target_adaptive_convergence}(i) holds with $a_{m, N} =O( \kappa\sqrt{{K\log(mK/\tau)}/{N}})$.

	Similarly, by Eq.~(\ref{equation:upper_bound_of_b}), we have
	\begin{align*}
		 & \sup_{j \in [m], \vq \in \triangle_{K}} \abs{\Lcal_{j,\targetx}(\vq) - \Lcal_{j}(\vq)} \\
		\leq & 2 \sup_{j \in [m], \vq \in \triangle_{K}} \abs{ \vq^\top (\bar{\vb}_{j} - \hat{\vb}_{j})} \\
		\leq & 2 \sup_{j \in [m], k \in [K]}\Bigabs{\hat{b}_{j,k} - \bar{b}_{j,k}} \\
		= & O \biggl(\kappa \sqrt{\frac{\log(mK/\tau)}{N}}\biggr),
	\end{align*}
	This ensures that Assumption~\ref{ass:target_adaptive_convergence}(ii) holds with the same $a_{m,N} = O(\kappa \sqrt{{K\log(mK/\tau)}/{N}})$.
    
	{\noindent\bf Verify Assumption~\ref{ass:adaptive_moment}.}
    Lemma~\ref{lemma:upper_var} indicates that Assumptions~\ref{ass:adaptive_moment}(i)--(ii) are both satisfied with $\sigma^2_{n} = O((n \psi)^{-1})$.
    
    Using the above results, by Lemma~\ref{lemma:sup_of_adaptive_RB}(i), we have
	\begin{align*}
		 & \sup_{\vq_1,\vq_2 \in \triangle_{K}} \frac{\abs{\mathrm{MWV}_{\epsilon_h}(\{\Ecal_j(\vq_1,\vq_2)\}_{j \in [m]}) - \Ecal_{\target}(\vq_1,\vq_2)}}{\max(\norm{\vq_1 - \vq^\ast}_{2} + \norm{\vq_2 - \vq^\ast}_{2},r_n)} \\
		\lesssim & \sqrt{\frac{\epsilon}{n \psi}} + \sqrt{\frac{\log(1/\tau) +K \log( n m K)}{n m \psi}} +\kappa\sqrt{\frac{K\log(mK/\tau)}{N}}.
	\end{align*}
	By Lemma~\ref{lemma:sup_of_adaptive_RB}(ii), we have \[\sup_{\vq \in \triangle_{K}} \abs{\hat{\Lcal}(\vq) - \Lcal_\target(\vq)} \lesssim  \sqrt{\frac{\epsilon}{n \psi}} + \sqrt{\frac{\log(1/\tau) +K \log( n m K)}{n m \psi}} +\kappa\sqrt{\frac{K\log(mK/\tau)}{N}}.\]
\end{proof}

\subsection{Proof of Theorem~\ref{theorem:robust_excess_mwv}}\label{sec:proof_thm1}
\begin{proof}
	It is sufficient to verify the assumptions in Theorem~\ref{theorem:general_one_step_update} with the MMD loss and the MWV estimator.
	When choosing MMD as the distribution divergence, the population divergence in Eq.~(\ref{equation:source_population_MMD}) is $\Lcal_Q(\vq) = \vq^\top \mA \vq - 2 \vq^\top \vb$ and the empirical divergence in Eq.~(\ref{equation:MMD}) is $\Lcal_{j}(\vq) = \vq^\top \hat{\mA}_{j} \vq - 2 \vq^\top \hat{\vb}_{j}$.

	The empirical divergence $\Lcal_{j}(\vq)$ is differentiable because it is a quadratic function of $\vq$.
	The minimum eigenvalue $\lambda_{\mathrm{min}}(\mA) > 0$ because the kernel is characteristic and the conditional distributions, $\{\targetxcy: k \in [K]\}$, are linearly independent.
	Given that $\mA \vq^\ast = \vb$, we have
	\[\Lcal_\target(\vq) - \Lcal_\target(\vq^\ast) = (\vq -\vq^\ast)^\top \mA (\vq -\vq^\ast) \geq \lambda_{\mathrm{min}} (\mA)\norm{\vq - \vq^\ast}_{2}^2.\]
	This indicates that Assumption~\ref{ass:convex} is satisfied with $\lambda = \lambda_{\mathrm{min}} (\mA)$.

    By Lemma~\ref{lemma:uniform_bound_mwv}, we have Assumption~\ref{ass:uniform_converge_excess_risk} holds with 
    \[
        r_n = \sqrt{\frac{\epsilon}{n \psi}} + \sqrt{\frac{\log(1/\tau) +K \log( n m K)}{n m \psi}} +\kappa\sqrt{\frac{K\log(mK/\tau)}{N}}.
    \]

	Finally, Theorem~\ref{theorem:robust_excess_mwv} follows from the above results and Theorem~\ref{theorem:general_one_step_update}.
\end{proof}

\subsection{Proof of Theorem~\ref{theorem:robust_divergence_mwv}}\label{sec:robust_divergence_mwv}
\begin{proof}
	We prove Theorem~\ref{theorem:robust_divergence_mwv} by checking the assumptions mentioned in Theorem~\ref{theorem:general_robust_loss}.
	Assumption~\ref{ass:convex} has been verified in the proof of Theorem~\ref{theorem:robust_excess_mwv}.
    By Lemma~\ref{lemma:uniform_bound_mwv}, we have Assumption~\ref{ass:uniform_converge_divergence} holds with the same rate,
    \[
        r_n = \sqrt{\frac{\epsilon}{n \psi}} + \sqrt{\frac{\log(1/\tau) +K \log( n m K)}{n m \psi}} +\kappa\sqrt{\frac{K\log(mK/\tau)}{N}}.
    \]

	Therefore, Theorem~\ref{theorem:robust_divergence_mwv} follows from Theorem~\ref{theorem:general_robust_loss}.
\end{proof}

\subsection{Illustrating the Variance Reduction Phenomenon by the Linear Model}
\label{sec:linear_model_variance_reduction}
Denote $(x,y)$ as a sample drawn from a linear model $y = x^\top \betabold^\ast + \epsilon$, and $\ell(\betabold) = (y - x^\top \betabold)^2$ as the squared loss.
We show the variance reduction of $\ell(\betabold) - \ell(\betabold^\ast)$ compared to $\ell(\betabold)$.
\begin{align*}
	 & \Var[\ell(\betabold)] -\Var[\ell(\betabold) - \ell(\betabold^\ast)] \\
	= & \E[\ell^2(\betabold)] - \{\E[\ell(\betabold)]\}^2 - \E\{[\ell(\betabold) - \ell(\betabold^\ast)]^2\} + \{\E[\ell(\betabold) - \ell(\betabold^\ast)]\}^2 \\
	= & \E\{[2 \ell(\betabold) - \ell(\betabold^\ast)] \ell(\betabold^\ast)\} - \E[2 \ell(\betabold) - \ell(\betabold^\ast)] \E[\ell(\betabold^\ast)] \\
	= & \Cov(2 \ell(\betabold) - \ell(\betabold^\ast), \ell(\betabold^\ast)) \\
	= & \Cov(\epsilon^2, \epsilon^2) = \Var(\epsilon^2).
\end{align*}
Therefore, the variance of the residual $\epsilon$ is reduced when $\ell(\betabold^\ast)$ is subtracted from the loss $\ell(\betabold)$.

\section{Additional Discussion on Label Shift Adaption}

\textbf{Label shift:} 
Estimating the target class proportion $\vq^\ast$ is inherently difficult due to the absence of labels from the target domain.
The initial approaches, as discussed in \cite{zhang2013domain}, \cite{iyer2014maximum}, and \cite{gong2016domain}, focus on estimating class proportions by aligning the marginal distribution $\targetx$ with a mixed distribution formed from the conditional distributions $\{P_{X|Y = k}: k \in [K]\}$.
This is achieved through the optimization task of a mixture proportion $\vq$ that lies within the $K$-dimensional probability simplex $\triangle_{K}$, as expressed in the following minimization problem,
\begin{equation*}
	\argmin_{\vq \in \triangle_{K}}\Lcal\biggl(\targetx,\sum_{k = 1}^K q_k P_{X|Y = k}\biggr),
\end{equation*}
where $\Lcal$ is a divergence measure between distributions. Notably, the Maximum Mean Discrepancy is utilized in \cite{zhang2013domain} and \cite{iyer2014maximum}, while the Wasserstein distance is employed in \cite{redko2019optimal} and \cite{shui2021aggregating}. In scenarios characterized by label shift, the minimization of this objective yields the optimal solution $\vq^\ast$.

A different avenue of research, initiated by \cite{lipton2018detecting}, introduces a framework for estimating the class ratio $\vr =[\Pr_\target(Y = k)/\Pr_P(Y = k)]_{k\in \Ycal}$ using the following estimation equation,
\[
 \hat{\vq}_f = \hat{\mC}_f \vr,
\]
where $f$ is a pretrained classifier, $\hat{\vq}_f$ signifies the estimated class proportion based on the outputs of $f$ on unlabeled target data, and $\hat{\mC}_f$ represents the confusion matrix for $f$ derived from the source data. 
This methodology has been enhanced by \cite{azizzadenesheli2018regularized} to be more effective in small-sample contexts and further generalized by \cite{garg2020unified} to include a broader class of models $f$, extending beyond the classification tasks.

From a theoretical perspective, the works of \cite{iyer2014maximum}, \cite{lipton2018detecting}, \cite{azizzadenesheli2018regularized}, and \cite{garg2020unified} established error bounds for the estimated class ratio/proportion. These bounds are characterized by the term $\sqrt{1/\min(n,N)}$, where $n$ and $N$ represent the sample sizes of the source and target datasets, respectively. This formulation underscores the relationship between the reliability of the estimates and the availability of data from both domains, highlighting the implications of sample size on the accuracy of class proportion estimation.

\textbf{Generalized label shift:} 
Recognizing that the label shift assumption may impose limitations in certain real-world contexts, the framework of Generalized Label Shift (GLS) \citep{zhang2013domain,gong2016domain,combes2020domain} accommodates situations where both the marginal distribution of $Y$ and the conditional distribution of $X|Y = k$ can differ between the source and target domains. The GLS framework maintains the assumption that there exists a feature mapping $\phi(X)$ such that the conditional distribution $\phi(X)|Y = k$ remains invariant across domains.

Within this framework, both the representation learner and the importance weights are iteratively refined during the training process by minimizing discrepancies within the representation space. However, the joint modeling of the two types of shifts, along with the latent invariant representation framework, raises concerns regarding model identifiability. 
This is particularly relevant when the number of source domains is limited, as it can lead to challenges in accurately estimating the underlying distributions and ensuring the stability of the learned representations. 
Consequently, addressing the identifiability issue is crucial for ensuring robust and reliable performance in scenarios where shifts in both the marginal and conditional distributions are present.

In contrast to the generalized label shift framework, our methodology relaxes the label shift assumptions by allowing an $\epsilon$ proportion of the source domains to be arbitrarily corrupted. This adjustment enhances flexibility and robustness in addressing variations that may arise in practical applications, thereby improving the overall robustness and applicability of our framework.

\section{Additional Discussions on Robust Estimation Techniques}\label{sec:robust_estimation}
In this section, we discuss more choices of robust weighting functions.
Additionally, in Section~\ref{sec:lem_trunc}, we provide neat theoretical results for the truncation method to demonstrate that the properties in Section~\ref{sec:theory} still hold for the truncated mean estimator.
The theory of other robust estimation techniques can be similarly developed.

Truncated mean is a commonly used robust mean estimation method, which removes both the top $\epsilon_h$ proportion and the bottom $\epsilon_h$ proportion of points, aiming to exclude the influence of extreme values or outliers.
The truncated mean estimator can be expressed as
\begin{equation}\label{equation:robust_estimation_trunc}
	\mathrm{Trunc}_{\epsilon_h}(\{z_j\}_{j \in [m]}) = \frac{1}{(1-2\epsilon_h)m} \sum_{j = \epsilon_h m + 1}^{(1-\epsilon_h)m} z_{(j)},
\end{equation}
where $z_{(j)}$ is the $j$-th smallest value in $\{z_j\}_{j \in [m]}$. Without loss of generality, assume that $\epsilon_h m$ is an integer.
The corresponding robust weighting function can be expressed as
\begin{equation}\label{equation:robust_weight_function_trunc}
	\RW_{\epsilon_h} (\{z_j\}_{j \in [m]}) = (w_j)_{j \in [m]} = \Bigl(\{(1-2\epsilon_h)m\}^{-1} \indicator_{\{ z_{( \epsilon_h m + 1)} \leq z_{j} \leq z_{( (1-\epsilon_h) m )}\}}\Bigr)_{j \in [m]},
\end{equation}
and is referred to as truncated weighting function.
The median can be viewed as a special case of truncated mean by setting $\epsilon_h = \lfloor \frac{n-1}{2} \rfloor/n$, which is considered in \cite{rousseeuw1984least}, with a breakdown point of nearly $50\%$.

Other robust methods can also be incorporated into our robust weighted framework, such as least trimmed squares \citep{rousseeuw1984least} and the median-of-means estimator \citep{devroye2016subgaussian}.
The least trimmed squares method removes the largest $\epsilon_h$ proportion of the losses, and the corresponding robust weight can be expressed as $\{w_j = \{(1-\epsilon_h)m\}^{-1} \indicator{\{z_{j} \leq z_{( (1-\epsilon_h) m )}\}}\}_{j \in [m]}$.
The median-of-means method evenly divides $\{z_j:j \in [m]\}$ into $L$ groups $\{G_{l}\}_{l \in [L]}$, computes the average $\abs{G_l}^{-1} \sum_{z \in G_l} z$ within each group, and then calculates the median of group averages.
The corresponding robust weight can be written as $\{w_j = m^{-1} L \indicator{\{z_{j} \in G_{med}\}}\}_{j \in [m]}$, where $G_{med}$ represents the group associated with the median of the group averages.

\subsection{Properties of the Truncated Mean Estimator}\label{sec:lem_trunc}
In this section, we present the properties of the truncated mean estimator in Eq.~(\ref{equation:robust_estimation_trunc}).
To apply Theorems~\ref{theorem:general_one_step_update} and \ref{theorem:general_robust_loss}, the main focus of this subsection is to prove Lemma~\ref{lemma:adaptive_trunc_convergence_rate}, which provides the necessary bounds as outlined in Assumptions~\ref{ass:uniform_converge_excess_risk} and \ref{ass:uniform_converge_divergence} when the truncated weighting function is used as $\RW$.
Based on these bounds, we draw parallel results to Theorems~\ref{theorem:robust_excess_mwv} and \ref{theorem:robust_divergence_mwv}.

We begin by presenting several properties of the truncated mean estimator, which are derived in Lemmas~\ref{lemma:truncation} and \ref{lemma:target_convergence}.

\begin{lemma}\label{lemma:truncation}
	Let $\Scal = \{x_i\}_{i \in [m]}$ be a set of $m$ random variables such that $\E[x_i] = \mu$ and $\max_{i \in [m]} \Var(x_i) \leq \sigma^2 < \infty$. $ \Tcal = (\Scal \backslash \Scal_r ) \cup \Scal_b = \{x_i^\prime\}_{i \in [m]} $ is the set obtained by arbitrarily removing an $ \epsilon $-proportion of points $ \Scal_r $ from $ \Scal $ and replacing them with another set $ \Scal_b $ of the same size.
	For $ 0<\tau<1 $, by setting $\epsilon_h = (1 + c) \epsilon + C \log(1/\tau) / m$ for some constant $C>0$ and any constant $c>0$, it follows that
	\[\abs{\operatorname{Trunc}_{\epsilon_h}(\{x^\prime_i\}_{i \in [m]}) - \mu} \lesssim \sigma \sqrt{\epsilon} + \sigma \sqrt{\frac{\log(1/\tau)}{m}},\]
    with probability at least $ 1-\tau$.
\end{lemma}

\begin{proof}
	Without loss of generality, we assume $\mu = 0$.
	Similar to the derivation of Eq.~(\ref{eq:num_of_removed_points}), with probability at least $1-\tau/2$, we have
	\[
		|\Scal^\prime| \geq (1-\epsilon^\prime) m,
	\]
	where $\epsilon^\prime = C \log(1/\tau)/m + c\epsilon$ for some constant $C$ and any constant $c>0$, and $\Scal^\prime = \{x_i: |x_i| < \sigma/\sqrt{\epsilon^\prime}\}$.	Let $z_i = x_i \mathbbm{1}_{\{\abs{x_i} \leq \sigma/\sqrt{\epsilon^\prime}\}}$.
	Given that $ \Var(z_i) \leq \sigma^2 $ and $|z_i|\leq \sigma / \sqrt{\epsilon^\prime}$, we can apply Bennett’s inequality \citep[see][Theorem 2.9.2]{vershynin2020high} to derive
	\begin{align*}
		\Pr\Bigl(\Bigl| \sum_{i \in [m]} (z_i - \E[z_i])\Bigr| > 2 m\sigma \sqrt{\epsilon^\prime}\Bigr) \leq \tau/2.
	\end{align*}
	The difference between the means of $x_i$ and $z_i$ can be bounded using Hölder's inequality via
	\begin{equation*}
		\begin{aligned}
			\bigabs{\E[x_i - z_i ]} = \bigl|\E[x_i \indicator_{\{\abs{x_i} > \sigma/\sqrt{\epsilon^\prime}\}}]\bigr|
			\leq \sigma \sqrt{\Pr\bigl(\abs{x_i} > \sigma/\sqrt{\epsilon^\prime}\bigr)}
			\leq \sigma \sqrt{\epsilon^\prime} ,
		\end{aligned}
	\end{equation*}
	as $ \Pr(\abs{x_i} > \sigma/\sqrt{\epsilon^\prime}) \leq \epsilon^\prime $ by Chebyshev's inequality.
	Combining the above inequalities, we obtain
	\begin{equation*}
		\Big|\sum_{x_i \in \Scal^\prime} x_i \Big| = \Big| \sum_{i \in [m]} z_i \Big| \leq 3 m\sigma \sqrt{\epsilon^\prime}.
	\end{equation*}

	Let $\Tcal^\prime$ be the set obtained by applying $ \epsilon_h $-truncation to $ \Tcal $.
	If $\epsilon_h \geq \epsilon+\epsilon^\prime $, then we can decompose $\Tcal^\prime$ into two subsets: $\Scal^\prime \cap \Tcal^\prime$ and $\Scal_b \cap E \cap \Tcal^\prime$, where $ E $ is the interval$ [-\sigma/\sqrt{\epsilon^\prime},\sigma/\sqrt{\epsilon^\prime}] $.
	Since $\Scal^\prime \subset E$ and $\abs{\Scal^\prime \setminus \Tcal^\prime} \leq 3\epsilon_h m$, we obtain
	\begin{align*}
		\Big|\sum_{x_i \in \Scal^\prime \cap \Tcal^\prime} x_i\Big| & \leq \Big|\sum_{x_i \in \Scal^\prime} x_i\Big| + \Big|\sum_{x_i \in \Scal^\prime \setminus \Tcal^\prime} x_i\Big| \\
		 & \leq 3 m \sigma \sqrt{\epsilon^\prime} + 3\epsilon_h m \sigma/\sqrt{\epsilon^\prime} \\
		 & = O\Bigl(m \sigma \sqrt{\epsilon} + \sqrt{m \log(1/\tau)}\Bigr).
	\end{align*}
	Since $\abs{\Scal_b \cap E \cap \Tcal^\prime} \leq \epsilon m$, we have
	\begin{align*}
		\Big|\sum_{x_i \in \Scal_b \cap E \cap \Tcal^\prime} x_i\Big| \leq \epsilon m \sigma/\sqrt{\epsilon^\prime} = O(m \sigma \sqrt{\epsilon}).
	\end{align*}
	Finally, the desired conclusion is obtained by
	\begin{align*}
		\Big|\mathrm{Trunc}_{\epsilon_h}(\{x^\prime_i\}_{i \in [m]})\Big| & = \frac{1}{(1-2\epsilon_h)m} \Big|\sum_{x_i^\prime \in \Tcal^\prime} x_i^\prime\Big| \\
		 & \leq \frac{1}{(1-2\epsilon_h)m} \Big|\sum_{x_i \in \Scal^\prime \cap \Tcal^\prime} x_i\Big| + \frac{1}{(1-2\epsilon_h)m} \Big|\sum_{x_i \in \Scal_b \cap E \cap \Tcal^\prime} x_i\Big| \\
		 & = O\Bigl(\sigma \sqrt{\epsilon} + \sqrt{\frac{\log(1/\tau)}{m}}\Bigr),
	\end{align*}
	with probability at least $ 1-\tau$.
	From the proof, we can set $\epsilon_h = (1+c)\epsilon + C \log(1/\tau)/m$ for any positive constant $c$ and a sufficiently large constant $C$.
\end{proof}

\begin{lemma}\label{lemma:target_convergence}
	For $\{x_i\}_{i \in [m]}$, $\{z_i\}_{i \in [m]} \subset \real$, if $\abs{z_i - x_i} \leq \zeta$, then we have
	\[\abs{\mathrm{Trunc}_{\epsilon_h}(\{x_i\}_{i \in [m]}) - \mathrm{Trunc}_{\epsilon_h}(\{z_i\}_{i \in [m]})} \leq \zeta.\]
\end{lemma}
\begin{proof}
	To begin with, we list two properties of $\mathrm{Trunc}_{\epsilon_h}$.
	\begin{itemize}
		\item \textit{Monotonicity}: For $\{\xi_{i}\}_{i \in [m]}$ and $\{\xi^\prime_{i}\}_{i \in [m]}$ such that $\xi_{i} \leq \xi^\prime_{i}$ for all $i \in [m]$, it follows that $\mathrm{Trunc}_{\epsilon_h}(\{\xi_{i}\}_{i \in [m]}) \leq \mathrm{Trunc}_{\epsilon_h}(\{\xi^\prime_{i}\}_{i \in [m]})$.
		\item \textit{Additivity }: For $\{\xi_{i}\}_{i \in [m]}$ and any constant $t$, it follows that $\mathrm{Trunc}_{\epsilon_h}(\{\xi_{i} + t\}_{i \in [m]}) = \mathrm{Trunc}_{\epsilon_h}(\{\xi_{i}\}_{i \in [m]}) + t$.
	\end{itemize}
	By these two properties, we have,
	\begin{align*}
		\mathrm{Trunc}_{\epsilon_h}(\{x_i\}_{i \in [m]}) \leq \mathrm{Trunc}_{\epsilon_h}(\{z_i + \zeta\}_{i \in [m]}) = \mathrm{Trunc}_{\epsilon_h}(\{z_i \}_{i \in [m]})+ \zeta.
	\end{align*}
	Similarly,
	\begin{align*}
		\mathrm{Trunc}_{\epsilon_h}(\{z_i\}_{i \in [m]}) \leq \mathrm{Trunc}_{\epsilon_h}(\{x_i \}_{i \in [m]})+ \zeta.
	\end{align*}
	Combining the above two inequalities, we have the conclusion.
\end{proof}

Now, we prove the key results for the truncation mean estimator.
Lemma~\ref{lemma:adaptive_trunc_convergence_rate} below provides bounds in Assumptions~\ref{ass:uniform_converge_excess_risk} and \ref{ass:uniform_converge_divergence} for the truncated mean estimator under the same assumptions of Lemma~\ref{lemma:sup_of_adaptive_RB} in Section~\ref{sec:lem_mwv}.
\begin{lemma} \label{lemma:adaptive_trunc_convergence_rate}
	Choosing $\epsilon_h = (1 + c) \epsilon + C \log(1/\tau) / m$ for any constant $c>0$, and some constant $C > 0$, the following results hold.

	(i) Under Assumption~\ref{ass:lipschitz}, \ref{ass:target_adaptive_convergence}(i) and \ref{ass:adaptive_moment}(i), with probability at least $1-\tau$, it follows
	\begin{align*}
		\sup_{\vq_1,\vq_2 \in \triangle_{K}} & \frac{\abs{\mathrm{Trunc}_{\epsilon_h}(\{\Ecal_j(\vq_1,\vq_2)\}_{j \in [m]}) - \Ecal_{\target}(\vq_1,\vq_2)}}{\max(\norm{\vq_1 - \vq^\ast}_{2} + \norm{\vq_2 - \vq^\ast}_{2},r_n)} \lesssim r_n.
	\end{align*}

	(ii) Under Assumptions~\ref{ass:lipschitz},\ref{ass:target_adaptive_convergence}(ii) and \ref{ass:adaptive_moment}(ii), with probability at least $1-\tau$, it follows
	\[\sup_{\vq \in \triangle_{K}} \abs{\mathrm{Trunc}_{\epsilon_h}(\{\Lcal_{j}(\vq)\}_{j \in [m]}) - \Lcal_\target(\vq)} \lesssim r_n,\]
	where $r_n = \sigma_n \sqrt{\epsilon} + \sigma_n \sqrt{\frac{\log(1/\tau) + K \log (L m/\sigma_n) }{m}} + a_{m, N}$.
\end{lemma}

\begin{proof}
	We first prove Part (i).
	We can construct a $\gamma$-net of $\triangle^{K} \times \triangle^{K}$ with cardinality $ O(\gamma^{-2K}) $, where $\gamma$ will be determined later.
	For any $(\vq_1,\vq_2) \in \triangle^{K} \times \triangle^{K}$, there exist $(\vq_1^\prime , \vq_2^\prime)$ in this $\gamma$-net such that $\max(\norm{\vq_1 - \vq_1^\prime}_{2},\norm{\vq_2 - \vq_2^\prime}_{2}) \leq \gamma$.
	Let $\Zcal_1 = \{z_{j,\targetx}^\prime\}_{j \in [m]}$, $\Zcal_2 = \{z_{j,\targetx}\}_{j \in [m]}$ and $\Zcal_3 = \{z_{j}\}_{j \in [m]}$, where $z_{j,\targetx}^\prime = \Ecal_{j,\targetx}(\vq_1^\prime,\vq_2^\prime) - \Ecal_Q(\vq_1^\prime,\vq_2^\prime)$, $ z_{j,\targetx} = \Ecal_{j,\targetx}(\vq_1,\vq_2) - \Ecal_Q(\vq_1,\vq_2)$ and $ z_{j} = \Ecal_{j}(\vq_1,\vq_2) - \Ecal_Q(\vq_1,\vq_2)$, respectively.
	
	Under Assumption~\ref{ass:adaptive_moment}(i), we have $\E[z_{j,\targetx}^\prime] = 0$ and $\Var(z_{j,\targetx}^\prime) \leq \sigma_n^2 (\norm{\vq_1^\prime - \vq^\ast}_{2} + \norm{\vq_2^\prime - \vq^\ast}_{2})^2$. By applying Lemma~\ref{lemma:truncation}, we have
	\[
	\abs{\mathrm{Trunc}_{\epsilon_h}(\{z_{j,\targetx}^\prime\}_{j \in [m]})} \lesssim \delta_n (\norm{\vq_1^\prime - \vq^\ast}_{2} + \norm{\vq_2^\prime - \vq^\ast}_{2}),
	\]	
	with probability at least $1-\tilde{\tau}/2$ for a fixed $(\vq_1^\prime, \vq_2^\prime)$, where $\delta_n = O(\sigma_n \sqrt{\epsilon} + \sigma_n \sqrt{\log(1/\tilde{\tau})/m})$.

	Given the Lipschitz continuity of $\Lcal_{j,\targetx}(\vq)$, we have $ |z_{j,\targetx} - z_{j,\targetx}^\prime| \leq 2 L (\norm{\vq_1 - \vq_1^\prime}_{2} + \norm{\vq_2 - \vq_2^\prime}_{2})\leq 4L \gamma $.
	Under Assumption~\ref{ass:target_adaptive_convergence}(i), we also have $ |z_{j,\targetx} - z_{j}| \leq a_{m,N} (\norm{\vq_1 - \vq^\ast}_{2} + \norm{\vq_2 - \vq^\ast}_{2})$ with probability at least $1-\tau/2$ uniformly for all $(\vq_1, \vq_2)$.
	By setting $\gamma = L^{-1}\delta_n r_n$,
	we obtain
	\begin{align*}
		\abs{\mathrm{Trunc}_{\epsilon_h}(\{z_{j}\}_{j \in [m]})} & \leq  \abs{\mathrm{Trunc}_{\epsilon_h}(\{z_{j,\targetx}^\prime\}_{j \in [m]})} \\
		& \ \ \ + \abs{\mathrm{Trunc}_{\epsilon_h}(\{z_{j,\targetx}\}_{j \in [m]}) -
			{\mathrm{Trunc}_{\epsilon_h}(\{z_{j,\targetx}^\prime\}_{j \in [m]})}} \\
		& \ \ \ + \abs{\mathrm{Trunc}_{\epsilon_h}(\{z_{j}\}_{j \in [m]}) - \mathrm{Trunc}_{\epsilon_h}(\{z_{j,\targetx}\}_{j \in [m]})} \\
		& = O\Bigl((\norm{\vq_1^\prime - \vq^\ast}_{2} + \norm{\vq_2^\prime - \vq^\ast}_{2}) \delta_n + L \gamma \\
		& \ \ \ \ \ \ \ \ \ \ \ + a_{m,N} (\norm{\vq_1 - \vq^\ast}_{2} + \norm{\vq_2 - \vq^\ast}_{2})\Bigr) \\
		& = O\Bigl((\norm{\vq_1 - \vq^\ast}_{2} + \norm{\vq_2 - \vq^\ast}_{2} + \gamma) (\delta_n + a_{m, N}) + r_n \delta_n \Bigr) \\
		& = O\Bigl((\delta_n + a_{m, N}) \max(\norm{\vq_1 - \vq^\ast}_{2} + \norm{\vq_2 - \vq^\ast}_{2},r_n)\Bigr),
	\end{align*}
	where the first equality follows from Lemma~\ref{lemma:target_convergence}. This indicates that with probability at least $1 - (\tau + \tilde{\tau}) / 2$,
	\[\frac{\abs{\mathrm{Trunc}_{\epsilon_h}(\{\Ecal_j(\vq_1,\vq_2)\}_{j \in [m]}) - \Ecal_{\target}(\vq_1,\vq_2)}}	{\max(\norm{\vq_1 - \vq^\ast}_{2} + \norm{\vq_2 - \vq^\ast}_{2},r_n)} \lesssim \sigma_n \sqrt{\epsilon} + \sigma_n \sqrt{\frac{\log(1/\tilde{\tau})}{m}} + a_{m, N},\]
	for any $\vq_1, \vq_2$ such that $\max(\norm{\vq_1 - \vq_1^\prime}_{2},\norm{\vq_2 - \vq_2^\prime}_{2}) \leq \gamma$.
	
	By choosing $\tilde{\tau} = \gamma^{2K} \tau$ and taking the union bound over all the $(\vq_1^\prime, \vq_2^\prime)$ in the $\gamma$-net, we obtain
	\begin{align*}
		& \sup_{\vq_1, \vq_2 \in \triangle_{K}} \frac{\abs{\mathrm{Trunc}_{\epsilon_h}(\{\Ecal_j(\vq_1,\vq_2)\}_{j \in [m]}) - \Ecal_{\target}(\vq_1,\vq_2)}}	{\max(\norm{\vq_1 - \vq^\ast}_{2} + \norm{\vq_2 - \vq^\ast}_{2},r_n)} \\
		\lesssim & \sigma_n \sqrt{\epsilon} + \sigma_n \sqrt{\frac{\log(1/\tau) + K \log (L m/\sigma_n) }{m}} + a_{m, N},
	\end{align*}
	which holds with probability at least $1-\tau$.
	
	We then study Part (ii).
	Similarly, we can construct a $ \gamma $-net of $ \triangle_{K} $ with cardinality $ O(1/\gamma^K) $, where $ \gamma = L^{-1}(\sigma_n \sqrt{\epsilon} + \sigma_n \sqrt{\log(1/\tau)/m} )$.
	For any $ \vq \in \triangle_{K} $, there exists $\vq^\prime$ in this $ \gamma $-net such that $ \norm{\vq^\prime-\vq}_{2} \leq \gamma $.
	
	Let $\Zcal_1 = \{z_{j,\targetx}^\prime\}_{j \in [m]}$, $\Zcal_2 = \{z_{j,\targetx}\}_{j \in [m]}$ and $\Zcal_3 = \{z_{j}\}$, where $z_{j,\targetx}^\prime = \Lcal_{j,\targetx}(\vq^\prime) - \Lcal_{\target} (\vq^\prime)$, $ z_{j,\targetx} = \Lcal_{j,\targetx}(\vq) - \Lcal_Q(\vq)$ and $ z_{j} = \Lcal_{j}(\vq) - \Lcal_Q(\vq)$, respectively.
	Under Assumption~\ref{ass:adaptive_moment}(ii), we have $\E[z_{j,\targetx}^\prime] = 0$ and $\Var(z_{j,\targetx}^\prime) \leq \sigma_n^2$.
	By applying Lemma~\ref{lemma:truncation}, we have
	\[
	\abs{\mathrm{Trunc}_{\epsilon_h}(\{z_{j,\targetx}^\prime\}_{j \in [m]})} \lesssim \delta_n,
	\]
	with probability at least $1-\tilde{\tau}/2$ for a fixed $\vq^\prime$, where $\delta_n = O(\sigma_n \sqrt{\epsilon} + \sigma_n \sqrt{\log(1/\tilde{\tau})/m})$.
	
	Given the Lipschitz continuity of $\Lcal_{j,\targetx}(\vq)$, we have $ |z_{j,\targetx} - z_{j,\targetx}^\prime| \leq 2 L \norm{\vq - \vq^\prime}_{2} \leq 4L \gamma $.
	Under Assumption~\ref{ass:target_adaptive_convergence}(i), we also have $ |z_{j,\targetx} - z_{j}| \leq a_{m,N}$ with probability at least $1-\tau/2$ uniformly for all $\vq$.
	Therefore, we have
	\begin{align*}
		\abs{\mathrm{Trunc}_{\epsilon_h}(\{z_{j}\}_{j \in [m]})} & \leq  \abs{\mathrm{Trunc}_{\epsilon_h}(\{z_{j,\targetx}^\prime\}_{j \in [m]})} \\
		& \ \ \ + \abs{\mathrm{Trunc}_{\epsilon_h}(\{z_{j,\targetx}\}_{j \in [m]}) -
			{\mathrm{Trunc}_{\epsilon_h}(\{z_{j,\targetx}^\prime\}_{j \in [m]})}} \\
		& \ \ \ + \abs{\mathrm{Trunc}_{\epsilon_h}(\{z_{j}\}_{j \in [m]}) - \mathrm{Trunc}_{\epsilon_h}(\{z_{j,\targetx}\}_{j \in [m]})} \\
		& = O\Bigl( \delta_n + L \gamma + a_{m,N} \Bigr),
	\end{align*}
	where the equality follows from Lemma~\ref{lemma:target_convergence}. This indicates that with probability at least $1 - (\tau + \tilde{\tau}) / 2$,
	\[\abs{\mathrm{Trunc}_{\epsilon_h}(\{\Lcal_j(\vq)\}_{j \in [m]}) - \Lcal_{\target}(\vq)} \leq \sigma_n \sqrt{\epsilon} + \sigma_n \sqrt{\frac{\log(1/\tilde{\tau})}{m}} + a_{m, N},\]
	for any $\vq$ such that $\lVert \vq - \vq^\prime \rVert_2 \le \gamma$.
	By choosing $\tilde{\tau} = \gamma^{K} \tau$ and taking the union bound over all $\vq^\prime$ in the $\gamma$-net, we obtain that with probability at least $1-\tau$,
	\[\sup_{\vq \in \triangle_{K}} \abs{\mathrm{Trunc}_{\epsilon_h}(\{\Lcal_{j}(\vq)\}_{j \in [m]}) - \Lcal_\target(\vq)} \lesssim\sigma_n \sqrt{\epsilon} + \sigma_n \sqrt{\frac{\log(1/\tau) + K \log (Lm/\sigma_n) }{m}} + a_{m, N}.\]
\end{proof}

Finally, the following theorem provides parallel results to those in Section~\ref{sec:theory}, when the truncated weighting function in Eq.~(\ref{equation:robust_weight_function_trunc}) is chosen as $\RW$.

\begin{theorem}\label{theorem:trunc_thm}
	Consider using Eq.~(\ref{equation:robust_weight_function_trunc}) as the robust weighting function with $\epsilon_h$ slightly larger than $\epsilon$ and $\epsilon_h = O(\epsilon)$.
	Under the same assumptions as in Theorem~\ref{theorem:robust_excess_mwv} and Theorem~\ref{theorem:robust_divergence_mwv}, with probability at least $1-\tau$, the same error bounds hold with
	\[
        r_n = \sqrt{\frac{\epsilon}{n \psi}} + \sqrt{\frac{\log(1/\tau) +K \log( n m K)}{n m \psi}} +\kappa\sqrt{\frac{K\log(mK/\tau)}{N}}.
    \]
\end{theorem}
\begin{proof}
    The proof of this theorem is identical to that of Theorem~\ref{theorem:robust_excess_mwv} and Theorem~\ref{theorem:robust_divergence_mwv}, by replacing Lemma~\ref{lemma:sup_of_adaptive_RB} with Lemma~\ref{lemma:adaptive_trunc_convergence_rate}.
\end{proof}

\section{Additional Numerical Results}
In this section, we analyze the impact of the unlabeled sample size $N$ and the hyperparameter $\epsilon_h$ on our estimation.
In the experiments, we vary either $N$ or $\epsilon_h$ while keep all other settings consistent with Section~\ref{sec:synthetic_data} such that $(m,n,N,\epsilon,\epsilon_h) = (40,100,4000,0.2,0.2)$.
The MSE results are illustrated in Figure~\ref{figure:epsilon_h}.

\begin{figure}[thb]
	\centering
	\includegraphics[width=0.831\linewidth]{"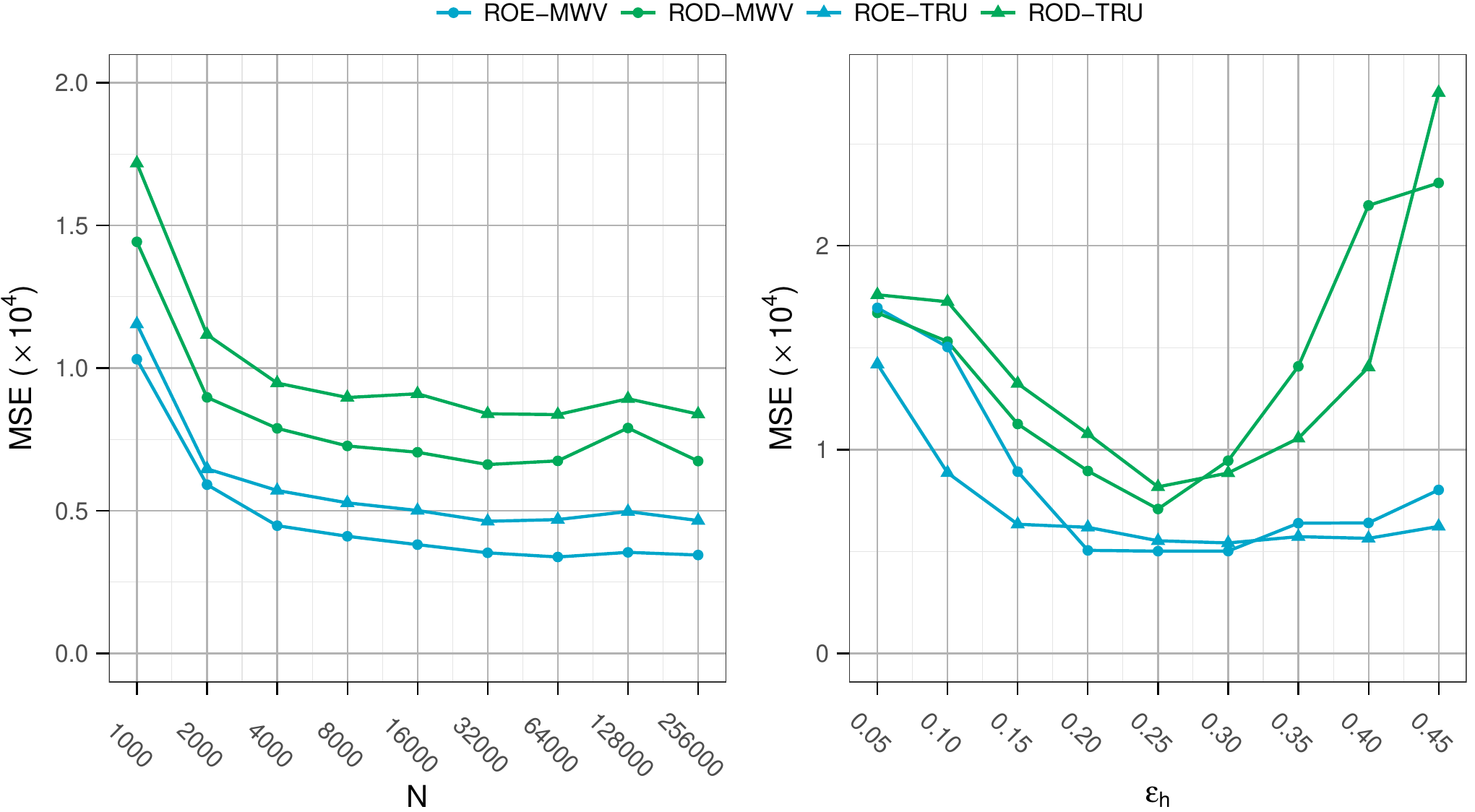"}
	\caption{The mean squared error (MSE) of the comparison methods evaluated under varying values of $N$ and $\epsilon_h$.}
	\label{figure:epsilon_h}
\end{figure}

The MSE decreases as $N$ increases.
However, when $N$ becomes sufficiently large (for example, 32,000), further increases in $N$ bring little improvement.
Thus, we recommend collecting $N = mn$ unlabeled target samples in practice, which ensures that the term involving $N$ in Corollary~\ref{corollary:general_multistep_update} becomes non-dominant.

When $\epsilon_h$ is slightly larger than $\epsilon$, the proposed method achieves optimal performance.
When $\epsilon_h$ is smaller than $\epsilon$, the failure to fully identify outlier sources leads to significant interference in the estimation results.
Conversely, when $\epsilon_h$ is much larger than $\epsilon$, the estimation performance slightly degrades due to the excessive removal of inlier sources, which causes a loss of efficiency.

\bibliographystyle{asa}
\bibliography{ref.bib}
\end{document}

%% file: symbols.tex
\ifcompact
\newcommand{\mypageset}{
    \geometry{
        paperwidth=6.34in,
        paperheight=9.17in,
        left=4mm,
        right=4mm,
        top=5mm,
        bottom=12mm,
    }
}
\mypageset
\fi

\def\vb{\mathbf{b}}

\def\vp{\mathbf{p}}
\def\vq{\mathbf{q}}
\def\vr{\mathbf{r}}

\def\vv{\mathbf{v}}

\def\mA{\mathbf{A}}

\def\mC{\mathbf{C}}
\def\mD{\mathbf{D}}

\def\Dcal{{\mathcal{D}}}
\def\Ecal{{\mathcal{E}}}

\def\Gcal{{\mathcal{G}}}

\def\Ical{{\mathcal{I}}}
\def\Jcal{{\mathcal{J}}}
\def\Kcal{{\mathcal{K}}}
\def\Lcal{{\mathcal{L}}}

\def\Ocal{{\mathcal{O}}}

\def\Rcal{{\mathcal{R}}}
\def\Scal{{\mathcal{S}}}
\def\Tcal{{\mathcal{T}}}

\def\Xcal{{\mathcal{X}}}
\def\Ycal{{\mathcal{Y}}}
\def\Zcal{{\mathcal{Z}}}

\def\Ebb{{\mathbb{E}}}

\def\Rbb{{\mathbb{R}}}

\def\sS{{\mathbb{S}}}

\newcommand{\indicator}{\mathbbm{1}}
\def\normal{\mathcal{N}}

\def\E{\mathbb{E}}

\def\real{\mathbb{R}}

\def\Var{\mathrm{Var}}

\def\Cov{\mathrm{Cov}}

\newcommand{\norm}[1]{\lVert #1\rVert}

\newcommand{\abs}[1]{|#1|}
\newcommand{\bigabs}[1]{\bigl| #1\bigr|}
\newcommand{\Bigabs}[1]{\Bigl| #1\Bigr|}

\makeatletter

\newcommand{\Rmnum}[1]{{\rm\expandafter\@slowromancap\romannumeral #1@}}
\makeatother